\documentclass[11pt,a4paper]{article}
\usepackage[truedimen,margin=32mm]{geometry} 

\usepackage{titlesec}
\titleformat*{\section}{\large\bfseries}
\titleformat*{\subsection}{\it}

\usepackage{mathrsfs}
\usepackage{booktabs}
\usepackage{amssymb}
\usepackage{amsmath}
\usepackage{mathtools}
\usepackage{ascmac}
\usepackage{amsthm}
\usepackage[pdftex]{graphicx}
\usepackage[pdftex]{color}
\usepackage{natbib}
\usepackage{setspace}
\usepackage{bm}
\usepackage{url}
\usepackage{color}
\usepackage{tabularx}
\usepackage{ulem}

\usepackage{algorithm}
\usepackage{algorithmic}

\newtheorem{corollary}{Corollary}
\newtheorem{lemma}{Lemma}
\newtheorem{proposition}{Proposition}

\theoremstyle{definition}

\newcommand{\mS}{\mathcal{S}}

\renewcommand{\hat}{\widehat}
\renewcommand{\tilde}{\widetilde}

\newcommand{\Pt}{\Phi_t}

\newcommand{\at}{\widetilde{a}}
\newcommand{\bt}{\widetilde{b}}

\usepackage{color}

\def\diag{\mathop{\rm diag}\nolimits}

\title{\bf Functional Horseshoe Smoothing \\for Functional Trend Estimation}
\author{}
\date{}

\begin{document}

\maketitle
\doublespacing

\vspace{-1.5cm}
\begin{center}
{\large 
Tomoya Wakayama$^{\ast}$\footnote{Corresponding author, Email: tom-w9@g.ecc.u-tokyo.ac.jp}
and Shonosuke Sugasawa$^{\dagger}$
}

\medskip

\medskip
\noindent
$^{\ast}$Graduate School of Economics, The University of Tokyo\\
$^{\dagger}$Center for Spatial Information Science, The University of Tokyo
\end{center}

\begin{abstract}
Due to developments in instruments and computers, functional observations are increasingly popular.
However, effective methodologies for flexibly estimating the underlying trends with valid uncertainty quantification for a sequence of functional data (e.g. functional time series) are still scarce.  
In this work, we develop a locally adaptive smoothing method, called functional horseshoe smoothing, by introducing a shrinkage prior to the general order of differences of functional variables.
This allows us to capture abrupt changes by making the most of the shrinkage capability and also to assess uncertainty by Bayesian inference.
The fully Bayesian framework also allows the selection of the number of basis functions via the posterior predictive loss.
Also, by taking advantage of the nature of functional data, this method is able to handle heterogeneously observed data without data augmentation.
Simulation studies and real data analysis demonstrate that, in comparison with alternative approaches, the proposed method achieves better smoothing accuracy and it has desirable properties.
\end{abstract}

\noindent%
\vfill

\newpage

\section{Introduction}

Owing to the remarkable development of measuring instruments and computers in recent years, it has become possible to obtain high-dimensional data in various fields. On the other hand, the analysis of such data using classical multivariate analysis requires a huge number of parameters, which is an obstacle in extracting valuable information from the data. A promising methodology to solve these problems is functional data analysis (FDA), which treats and analyzes high-dimensional data as a curve (function). Functional versions for many branches of statistics have been provided such as \cite{ramsay2004functional}, \cite{kokoszka2017introduction}, and \cite{horvath2012inference}, and it is a field that has been the focus of much research.

The traditional FDA approach for independent functional data has recently been extended to time series.
In fact, for functional time series data, the standard stationary model for multivariate data has been extended \citep[e.g.][]{besse2000autoregressive,klepsch2017innovations,klepsch2017prediction,hormann2013functional,gao2019high,hormann2015dynamic} and its theoretical properties have been extensively studied \citep[e.g.][]{bosq2000linear,aue2017estimating,spangenberg2013strictly,aue2017functional,kuhnert2020functional,cerovecki2019functional}.
However, in actual data such as GDP, the assumption of stationarity is not satisfied in many situations because the expected value varies from period to period. There are a few cases where the trend can be analyzed appropriately by existing methods.

Due to the stationarity of conventional FDA methods, they cannot capture rapid changes in trend estimation, but \cite{wakayama2021locally} solved this difficulty by developing a new type of lasso and proposing functional trend filtering. This new method can capture local changes while removing observation errors in the data. In other words, the method can clearly identify when structural changes occur in time series data. In order for the inferred results to be used for decision-making,  it is essential to evaluate the interpretability of the model and the uncertainty of the estimation. However, few methods for uncertainty evaluation in functional time series analysis have been developed \citep{petris2013bayesian,canale2016bayesian}.

In this work, we develop the other approach in a Bayesian framework, hoping to assess uncertainty and estimate trend accurately and flexibly as seen in univariate models \citep{faulkner2018locally}.
In the context of FDA, a shrinkage prior on the functional space has been introduced by \cite{shin2020functional}.
Utilizing a similar idea, we construct a locally adaptive smoother for functional data via the shrinkage prior. Because the priors in the model can be represented as scale mixture of normals, it is easy to implement with the Gibbs sampler, and minor extensions make heterogeneously observed data analysis possible. 
We derive posterior predictive loss based on \cite{gelfand1998model} to choose an appropriate number of basis functions.

We also discuss the theoretical justification for this approach. The essence of this method is that it removes noise while leaving the change points large.
The property of the prior that keeps the signal from shrinking is called "tail robustness". For time series data, the analysis of tail robustness is complicated, but instead, we have shown the properties proposed by \cite{okano2022locally}.
This is not an argument closed to functional data, but also justifies a locally adaptive method for finite-dimensional data \citep{faulkner2018locally,kakikawa2022bayesian}.

The remainder of the paper is structured as follows. 
Section \ref{sec:set} introduces the setting and model for trend estimation. Section \ref{sec:pos} gives the posterior computation algorithm. In Section \ref{sec:sel}, we present the way to select the number of basis functions. Section \ref{sec:tai} discusses the theoretical properties of the proposed prior and its posterior distribution. In Section \ref{sec:sim}, we investigate the performance of the proposed methods for homogeneously observed data and heterogeneously observed data. We apply our method to a real dataset in Section \ref{sec:app}. The contribution of the article is discussed in Section \ref{sec:dis}.

\section{Functional Horseshoe Smoothing}\label{sec:BTF}

\subsection{Settings and models}\label{sec:set}
Let $Y_1(\cdot),\ldots, Y_T(\cdot)$ be observed functional data on $\mathcal{S}\subset \mathbb{R}$, ordered as $t=1,\ldots,T$.
Suppose that we are interested in the mean function $Z_t(\cdot)\equiv {\rm E}[Y_t(\cdot)]$ a that may change smoothly or abruptly over $t$.  
To estimate $Z_t$, we employ the following measurement error model: 
$$
Y_t(s)=Z_t(s)+\varepsilon_t(s), \ \ \ \ \varepsilon_t(s)\sim N(0,\sigma^2), \ \ \ \ t=1,\ldots,T,\ \ \ \ s\in\mS,
$$
where $\varepsilon_t(s)$ are error terms independent over different values of $t$ and $s$, and $\sigma^2$ is an unknown variance.
Such measurement models are adopted in the context of Bayesian modeling of functional data \citep{yang2016smoothing,yang2017efficient}.

Let $\phi_1(\cdot),\ldots,\phi_L(\cdot)$ be basis functions on $\mathcal{S}$ (e.g. B-spline function) common over $t$.
We model $Z_t(\cdot)$ as
$$
Z_t(\cdot)=\sum_{\ell=1}^L b_{t\ell} \phi_{\ell}(\cdot), \ \ \ \ \ t=1,\ldots,T,
$$
where $\bm{b}_t=(b_{t1},\ldots,b_{tL})^\top$ is a vector of coefficients, and $L$ is the number of basis functions.  
Thus, heterogeneity of the mean function $Z_t(s)$ is characterized by the heterogeneous coefficients, $\bm{b}_t$.
The choice of $L$ controls the smoothness of the estimates of $Z_t(\cdot)$, and we later discuss a data-dependent selection of $L$.

Let $Y_t(s_{t1}),\ldots, Y_t(s_{tn_t})$ be discrete observations, where the $s_{t1},\ldots,s_{tn_t}$ are the observation points and $n_t$ is the number of discrete observations.
Note that we allow the number of sampling points and the sampled locations to be heterogeneous over $t$. 
Under the settings described above, the model for $\bm{y}_t=(Y_t(s_{t1}),\ldots,Y_t(s_{t{n_t}}))^\top$ is 
$$
\bm{y}_t=\Pt \bm{b}_t + \bm{\varepsilon}_t, \ \ \ \ \ \bm{\varepsilon}_t\sim N(\bm{0}, \sigma^2 I_{n_t}), \ \ \ \ t=1,\ldots,T, 
$$
where $\Pt$ is a $n_t\times L$ matrix whose $(i,\ell)$-element is $\phi_\ell(s_{ti})$.

Now, we consider prior distributions on $\bm{b}_t$. $\otimes$ denotes Kroneker product.
Let $\Delta_k$ be $k$th order forward difference operators defined as \begin{align*}
\Delta_k= \begin{cases}
    D^{(0)} & {\rm for} \,\,k=0 ,\\
    D^{(k)}\Delta_{k-1} & {\rm for} \,\,k\geq 1,
 \end{cases}
\end{align*}
where $D^{(k)}$ is the following $ (T-k-1)L \times (T-k)L$ matrix:
\begin{align*}
    D^{(k)} &=\left(
        \begin{array}{cccccc}
          1 & -1& 0&\ldots & 0 &0\\
          0 & 1 & -1&\ldots & 0 &0\\
          \vdots & \vdots&\vdots & \ddots & \vdots&\vdots \\
          0 & 0&0 & \ldots & 1&-1
        \end{array}
        \right) \otimes I_L.
\end{align*}
We then define $(\bm{\delta}_1,\ldots,\bm{\delta}_{T-k-1})^{\top}= \Delta_k (\bm{b}_1,\ldots,\bm{b}_T)^{\top}$ and consider a model for this. Although $\bm{\delta}_t$ depends on $k$, we write $\bm{\delta}_t$ rather than $\bm{\delta}_t^{(k)}$ for simplicity.
For example, if $k=0$, each $\bm{\delta}_t$ can be written as $\bm{\delta}_t=\bm{b}_{t}-\bm{b}_{t+1}$, for $t=1,\ldots,T-1$. 
Hence, if the vector $\bm{\delta}_t$ is shrunk toward the origin, the two adjacent coefficient vector $\bm{b}_{t+1}$ and $\bm{b}_{t}$ are identical, leading to the same mean functions for $Z_t(\cdot)$ and $Z_{t+1}(\cdot)$.
To encourage such structures, we introduce shrinkage priors on $\bm{\delta}_t$, which has a large mass around the origin, whereas the tail of the prior is sufficiently large to allow possible abrupt changes over $t$.
Note that under general $k$, shrinking $\bm{\delta}_t$ toward the origin can be regarded as the smoothing of $k+1$th order derivatives of $Z_t(\cdot)$ with respect to $t$. 
We introduce the following hierarchical prior for $\bm{\delta}_t$:
\begin{equation}\label{prior}
\begin{split}
\bm{\delta}_t|\lambda_t,\tau,\sigma \sim N(\bm{0},\sigma^2\tau^2\lambda_t^2(\Pt^{\top}\Pt)^{-1}), \ \ \ 
\lambda_t\sim C^{+}(0, 1),   \ \ \  \tau\sim C^{+}(0, 1), 
\end{split}
\end{equation}
where $C^{+}(0, 1)$ denotes the standard half-Cauchy distribution.
The conditional prior of $\bm{\delta}_t$ has a form of $g$-prior \citep{zellner1986assessing}, and a hierarchical prior similar to (\ref{prior}) is adopted in \cite{shin2020functional} in the context of linear regression models. 
Here $\lambda_t$ is a local shrinkage parameter that controls the amount of shrinkage, and $\lambda_t$ is common to all components of $\bm{\delta}_t$ so that the vector $\bm{\delta}_t$ can be simultaneously shrunk toward the origin.
The use of the half-Cauchy for the local parameter $\lambda_t$ leads to a multivariate horseshoe-like prior for $\bm{\delta}_t$. The prior distribution has favorable properties, given in Section \ref{sec:tai}, and we make the most of them to handle sparsity.

\subsection{Posterior computation algorithm }\label{sec:pos}
The joint posterior distribution is given by 
\begin{align*}
\pi(\sigma^2)\pi(\tau^2)\prod_{t=1}^Tp(\bm{y}_t;\Pt\bm{b}_t,\sigma^2 I_n)
\prod_{t=1}^{T-k-1}p(\bm{\delta}_t; \bm{0}, \sigma^2\tau^2\lambda_t^2(\Pt^\top\Pt)^{-1})\pi(\lambda_t),
\end{align*}
where $\pi(\sigma^2)$, $\pi(\tau^2)$ and $\pi(\lambda_t)$ are prior distributions for $\sigma^2, \tau^2$ and $\lambda_t^2$, respectively. 
The priors of $\tau^2$ and $\lambda_t^2$ are defined in (\ref{prior}), and we use the conjugate prior, $\sigma^2\sim {\rm IG}(a_{\sigma}, b_{\sigma})$ given hyperparameters $a_{\sigma}$ and $b_{\sigma}$.
Using the data augmentation technique of the horseshoe prior \citep[e.g.][]{makalic2015simple}, sampling from the joint posterior can be carried out by a simple Gibbs sampling described as follows:

\begin{itemize}
\item[-]
(Sampling from $\sigma^2$) \ \ 
The full conditional distribution of $\sigma^2$ is ${\rm IG}(\at_{\sigma}, \bt_{\sigma})$, where 
\begin{align*}
&\at_\sigma=a_{\sigma}+\frac12L(T-k-1)+\frac12nT, \\ 
&\bt_\sigma=b_{\sigma}
+\frac12 \sum_{t=1}^T(\bm{y}_t - \Pt\bm{b}_t)^\top
(\bm{y}_t - \Pt\bm{b}_t) 
+\frac1{2\tau^2} \sum_{t=1}^{T-k-1}\frac{\bm{\delta}_t^\top \Pt^\top \Pt \bm{\delta}_t}{\lambda_t^2}.
\end{align*}

\item[-]
(Sampling from $\tau^2$) \ \ The full conditional distribution of $\tau^2$ is
$$
\mathrm{IG}\left(\frac{L(T-k-1)+1}{2},\, \frac{1}{\xi}+\sum_{t=1}^{T-k-1}\frac{\bm{\delta}_t^\top \Pt^\top \Pt \bm{\delta}_t}{2\sigma^2\lambda_t^2}  \right),
$$
where $\xi$ is an auxiliary variable whose full conditional distribution is 
$\mathrm{IG}\left(1,\, 1+1/\tau^2  \right).$

\item[-]
(Sampling from $\lambda_t^2$) \ \ The full conditional distribution of $\lambda_t^2$ is 
$$
\mathrm{IG}\left(\frac{L+1}{2},\, \frac{1}{\nu_t}+\frac{\bm{\delta}_t^\top \Pt^\top \Pt \bm{\delta}_t}{2\tau^2\sigma^2}\right),
$$
where $\nu_t$ is an auxiliary variable whose full conditional distribution is $\mathrm{IG}\left(1, 1 + 1/\lambda_t^2\right)$.

\item[-]
(Sampling from $\bm{b}_t$) \ \ 
The full conditional distribution of $\bm{b}_t$ is of the form, $N(\bm{\mu}_t, c_t(\Pt^\top \Pt)^{-1})$, where the specific forms of $\bm{\mu}_t$ and $c_t$ are dependent on $k$, the order of difference. 
The detailed expressions under $k=0$ and $k=1$ in the Appendix 1. 
\end{itemize}

As shown above, all the full conditional distributions are familiar forms; thereby the posterior computation can be efficiently carried out. 
Given the posterior samples of $\bm{b}_t$, the posterior samples of $Z_t(s)$ at arbitrary location $s\in \mathcal{S}$ can be generated, which gives point estimate (e.g. posterior mean) and interval estimation (e.g. $95\%$ credible interval).

\subsection{Selection of the number of basis functions}\label{sec:sel}
In practice, the specification of the number of basis functions, $L$, is an important task.
If $L$ is smaller than necessary, the basis function approximation gives over-smoothed results.
On the other hand, the estimation results can be inefficient if $L$ is larger than necessary.
We suggest adopting a model selection criterion to select $L$ in a data dependent manner. 
We here use posterior predictive loss (PPL) proposed by \cite{gelfand1998model}.

To clarify the number of basis used in the estimation, we write $\bm{b}_t(L)$ and $\Pt(L)$ rather $\bm{b}_t$ and $\Pt$.
Given $\bm{b}_t(L)$, the conditional distribution of $\bm{y}_t$ is $N\left(\Pt(L)\bm{b}_t(L), \sigma^2I_n \right)$. 
We then define the PPL as 
\begin{align*}
{\rm PPL}(L) &= \frac{T}{T+1}\sum_{t=1}^T \left\{ \bm{y}_t-\Pt(L){\rm E}_p[\bm{b}_t(L)]\right\}^\top
\left\{ \bm{y}_t-\Pt(L){\rm E}_p[\bm{b}_t(L)]\right\}\\
& \ \ \ \ +nT{\rm E}_p[\sigma^2] + \sum_{t=1}^T\mathrm{tr} (\Pt(L){\rm Cov}_p(\bm{b}_t(L))\Pt(L)^\top),
\end{align*}
where ${\rm E}_p$ and ${\rm Cov}_p$ are the expectation and covariance with respect to the posterior distribution. 
We choose the number of basis functions such that the criterion ${\rm PPL}(L)$ is minimized. 
The order of differences, $k$, can also be selected by PPL.

\subsection{Extension to irregular grids}
We here consider an extension to the proposed smoothing techniques under irregularly spaced functional data. 
Let $Z_{t_1}(\cdot),Z_{t_2}(\cdot),\ldots,Z_{t_n}(\cdot)$ be a sequence of functional random variables indexed by $t_i$. 
This situation requires that distance information be incorporated into the model. 
Using the same basis expansion in Section~\ref{sec:set}, we introduce the following prior:  
\begin{align}\label{irregular}
\bm{b}_{t_i+h} - \bm{b}_{t_i}|\lambda_{t_i},\tau,\sigma \sim N(\bm{0},h\sigma^2\tau^2\lambda_{t_i}^2(\Phi_{t_i}^{\top}\Phi_{t_i})^{-1}), \ \ \ \ h\geq 0
\end{align}
where we use the same prior for $\lambda_t^2$.
Although the prior formulation (\ref{irregular}) corresponds to the extension under first order difference, the second order cases can also be extended to the irregular grids, following the argument given by \cite{lindgren2008second}.

\section{Theoretical properties of the model}\label{sec:tai}
This section presents the theoretical properties of the prior distribution and its periphery. Their proofs are described in Appendix 2.

In Section \ref{sec:set}, we formulated a prior as (\ref{prior}) and now investigate it in detail.
Its marginal prior of $\lambda_t$ is 
$$
\pi(\bm{\delta}_t \mid \tau,\sigma) \propto
\int_{0}^{\infty}\frac{1}{\lambda^L_t(1+\lambda_t^2)}\exp\left\{ -\frac{1}{2\sigma^2\tau^2\lambda_t^2}\bm{\delta}^{\top}_t \Pt^{\top} \Pt \bm{\delta}_t \right\}d\lambda_t.
$$
Then, notable properties of the marginal prior are given by the following proposition.
\begin{proposition}\label{prop:pri1}
\begin{align*}
(&i)\ \ \ \ \ \ \  \pi(\bm{\delta}_t \mid \tau,\sigma)\rightarrow \infty \ \ \ \ as \ \ \ \ \bm{\delta}_t \rightarrow 0.\\
(i&i)\ \ \ \ \ \ \  \pi(\bm{\delta}_t \mid \tau,\sigma) = O \left(\|\Pt \bm{\delta}_t\|_2^{-L-1}\right)
\end{align*}
\end{proposition}$(i)$ implies that, for given $\tau,\sigma^2$ and $L$, the density diverges at the origin $\bm{\delta}_t=\bm{0}$ like the original horseshoe prior \citep{carvalho2009handling,carvalho2010horseshoe}.
Conspicuously, this property strongly shrinks the trivial noises toward zero at posterior inference. On the other hand, $(ii)$ suggests that the tail decay of the marginal prior is slow. The random variables from the prior are expected to take large values with greater probability owing to the heavy tail. These critical features of the prior distribution contribute to handling sparsity.

Next, we consider the posterior mean deduced from the prior. For simplicity, we focus on $k=0$.
In this case, the model can be rewritten as
\begin{align*}
\bm{z}_t &\equiv \bm{y}_{t+1}-\bm{y}_{t}\sim N(\Pt\bm{\delta}_t, 2\sigma^2I_n), \ \ \  
\bm{\delta}_t \equiv  \bm{b}_{t+1}-\bm{b}_{t}\sim N(\bm{0},\sigma^2\lambda_t^2\tau^2(\Pt^\top \Pt)^{-1}),
\end{align*}
so that the model is defined for the observed value of difference $\bm{z}_t$.

\begin{proposition}\label{prop:tail}
The posterior mean of the model is weakly tail robust, that is,
\begin{align*}
\frac{|{\rm E}[\Phi_{t^*}\bm{\delta}_{t^*}|\bm{z}] - \bm{z}_{t^*}|}{\|\bm{z}_{t^*}\|_2}\to \bm{0} \ \ \ \ \mathrm{as}\ \ \ \  \|\bm{z}_{t^*}\|_2 \to \infty \ \ \ \ \mathrm{for}\ \ \mathrm{any}\ \ t^* \in \{1,...,T-1\}
\end{align*}
\end{proposition}
This claim implies that the difference between the posterior expectation and the original observation are relatively subtle when $\|\bm{z}_{t^*}\|_2$ is large.
This property is weaker than the tail robustness \citep{carvalho2010horseshoe}.
In our setting, the dependencies between the data make it challenging to analyze the tail robustness, but still, weak tail robustness, which implies the signal is preserved in the posterior analysis without shrinkage, holds.
The property is derived from the fact that the prior has considerable mass on the tail.
The use of $C^{+}(0,1)$ for $\lambda_t$ is motivated by this argument.

The tail robust-related properties of time series have not been shown in ordinary multivariate analysis as well as in functional data.
The result of this theorem also applies to ordinary multivariate analysis by ignoring $\Phi_t$, and this is important in the context of shrinkage estimation.

\section{Simulation Studies}\label{sec:sim}

\subsection{Simulation settings}
\label{sec:sim:pro}
We evaluate the performance of the proposed and existing methods through simulation studies. 
For $t=1,...,T(=50)$ and the domain $\mathcal{X}=[1,n]$ with $n=120$, we prepared the following four scenarios as the true trend function $\beta_t(x)$:
\begin{align*}
&\text{(1) Constant:}  \ \ \beta_t(x)=f_1(x),\\
&\text{(2) Smooth:} \ \ \beta_t(x)=f_1(x)\sin\frac{t+x}{5},\\ 
&\text{(3) Piecewise constant:} \ \ 
\beta_t(x)=\sum_{i=1}^5 f_i(x)\mathbb{I}_{\{10(i-1) < t\leq 10i\}} ,\\
&\text{(4) Varying\ smoothness:} \ \ 
\beta_t(x)=f_1(x)+20\big\{\sin\big(\frac{4t}{n}-2\big)+2\exp\big(-30\big(\frac{4t}{n}-2\big)^2\big)\big\},
\end{align*}
where $f_i$ ($i=1,\ldots,5$) is a sample path of the Gaussian process associated to RBF kernel $k_i(x_1,x_2) = \theta_i^2 \exp(-\|x_1-x_2\|^2/2\theta_i^2)$ with a hyper-parameter $\theta_i$. We set $\theta_i= 30, 20, 35, 25, 30$ for $i=1,2,\ldots,5$.

The observed data were generated by adding noises from $N(0,5^2)$ to trend functions at equally spaced $H=120$ points of $x$, namely, $x\in \{1,2,\ldots,H\}$.
Figure \ref{data_3d} shows how the trends change over time.
\begin{figure}[tb]
  \begin{center}
  \includegraphics[width=14cm]{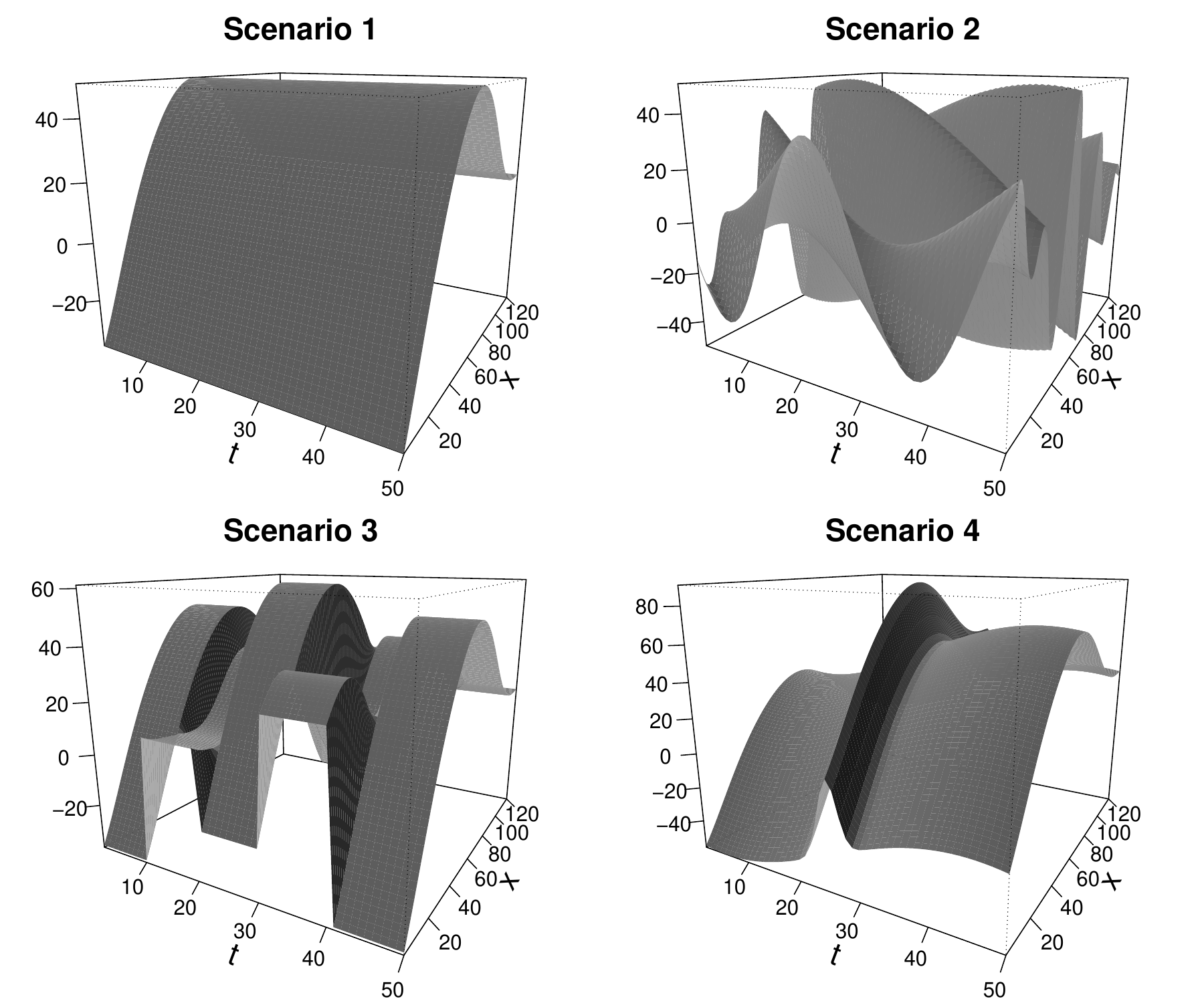}
  \end{center}
  \caption{Each surface represents a three-dimensional plot of the true trend.\label{data_3d}}
\end{figure}

In scenario 1, we investigate whether the proposed methods are able to discover the trend is constant over time even in the presence of noise. 
Scenario 2 checks the ability of the methods to extract a continuous curve from noisy data. 
Scenario 3 reveals the ability of the methods to detect abrupt changes between intermittent straight horizons, that is, discontinuity points.
In scenario 4, we examine whether the methods can capture periods when the smoothness of the trend changes significantly.

\subsection{Homogeneously observed data}
\label{sec:sim1}
We first considered using the full dataset generated according to the method presented in the previous section, which we call {\it homogeneously observed data}.
For the simulated data, we applied the following methods: 
\begin{itemize}
\item[-]
FHS: functional horseshoe smoothing. 

\item[-]
FLS: an artificial alternative method using the Laplace-like prior, that is, $\lambda_t^2\sim {\rm Exp}(1)$. 

\item[-]
B-spline: a curve fitting method using B-spline function, which is also used as the basis of the above two methods.

\item[-]
FTF: Functional Trend Filtering developed by \cite{wakayama2021locally}.

\item[-]
BART: Bayesian additive regression trees developed by \cite{chipman2010bart}.

\end{itemize}

The motivation for using the FLS method is to address the importance of the half-Cauchy prior for the local parameter $\lambda_t$ as discussed in Section \ref{sec:BTF}.
The purpose of using BART is to compare FHS with existing flexible methods. In fact, this time-dependent functional analysis can be reframed as a bivariate regression problem, for which BART can be applied.
Also, to compare FHS with locally adaptive frequentist methods of functional time series analysis, we applied FTF.

For the Bayesian methods, we used 3000 posterior draws after discarding 3000 burn-in samples.
In the two methods, FHS and FLS, the optimal number of the basis functions, $L$, and the order of difference, $k$, is selected via the PPL criterion among candidates $L\in \{5,9,13,17,21,25\}$ and $k\in \{0,1\}$.

For the evaluation of point estimates, we adopted the following criterion:
\begin{itemize}
\item[-]
Mean absolute deviation (MAD): Difference between the posterior medians and the true values, defined as  
$$
\mathrm{MAD}=\frac{1}{HT} \sum_{x=1}^{H}\sum_{t=1}^T  \bigl|\hat{\beta}_t(x)-\beta_t(x)\bigl|.
$$

\end{itemize}
Moreover, we used the following two criteria to evaluate $95\%$ credible intervals obtained from FHS and FLS:
\begin{itemize}
\item[-]
Mean credible interval width (MCIW): The width of intervals, define as 
$$
\mathrm{MCIW}=\frac{1}{HT} \sum_{x=1}^{H}\sum_{t=1}^T  \hat{\beta}^{97.5}_t(x)-\hat{\beta}^{2.5}_t(x),
$$
where $\hat{\beta}^{97.5}_t(x)$ and $\hat{\beta}^{2.5}_t(x)$ correspond to $97.5$ and $2.5$ percentiles of posterior distribution for $\beta_t(x)$.
    
\item[-]
Mean coverage (MC): The coverage accuracy of the credible interval, defined as 
$$
\mathrm{MC}=\frac{1}{HT} \sum_{x=1}^{H}\sum_{t=1}^{T} \mathbb{I}_{\{\hat{\beta}^{97.5}_t(x)>\beta_t(x) >\hat{\beta}^{2.5}_t(x)\}}.
$$
\end{itemize}

We repeated the simulations 150 times, and the averages across simulations are presented in Table \ref{tab:full}.

\begin{table}[p]
\caption{MAD (mean absolute deviation), MCIW (mean credible interval width), MC (mean coverage) and the number $L$ of basis for FHS (the estimator using the horseshoe-like prior), FLS (the estimator based on the Laplace-like prior), Spline1 (B-spline estimator with the same basis as FHS) and Spline2 (B-spline estimator with the same basis as FLS) for each scenario.\label{tab:full}} 
\centering
\medskip
\begin{tabular}{cccccccccccc} 
\toprule
Scenario & Method &  & MAD & MCIW & MC & $L$\\
\midrule
1&FHS & & 0.538 & 3.321 &  0.981 & 17.9\\  
&FLS & & 0.673 & 4.374 &  0.990 & 24.9\\ 
&Spline1 & & 1.495 & - &  - & 17.9\\  
&Spline2 & & 1.127 & - & - & 24.9\\
&BART & & 0.656 & 3.187 &  0.934 & -\\ 
&FTF & & 0.490 & - & - & - \\
\midrule
2&FHS& & 0.961 & 5.214 &  0.966 & 23.4\\  
&FLS& & 0.940 & 4.590 & 0.941 & 24.9  \\  
&Spline1 & & 1.739 & -& - & 23.4\\  
&Spline2 & & 1.333 & - & - & 24.9\\
&BART & & 2.539 & 9.122 & 0.830 & -\\ 
&FTF & & 1.700 & - & - & - \\
\midrule
3&FHS& & 0.713 & 4.350 & 0.982 & 21.5 \\  
&FLS& & 3.701 & 6.307 & 0.592 & 9.24  \\ 
&Spline1 & & 1.560 & - &  - & 21.5\\  
&Spline2 & & 4.045 & - &  - & 9.24\\
&BART & & 1.217 & 5.061 &  0.891 & -\\ 
&FTF & & 2.245 & - & - & - \\
\midrule
4&FHS& & 0.814 & 3.988 &  0.937 & 19.0 \\  
&FLS& & 1.523 & 4.276 &  0.807 & 10.1 \\  
&Spline1 & & 1.542 & - &  -   & 19.0\\  
&Spline2 & & 1.868 & - &  - & 10.1\\
&BART & & 0.691 & 3.210 &  0.926 & -\\ 
&FTF & & 1.213 & - & - & - \\
\bottomrule
\end{tabular}
\end{table}

Overall, we found that FHS had better performance than the other method from Table \ref{tab:full}. The concrete discussion about estimation accuracy is as follows.
\begin{itemize}
    \setlength{\leftskip}{1.2cm}
    \item[Scenario 1:] Frequentist shrinkage methods generally have a stronger ability to shrink estimators to zero than Bayesian methods, and hence FTF yields the best results. However, the difference with FHS is subtle, and FHS's ability to reduce is also significant.
    \item[Scenario 2:] The function data was a complex and ever-changing signal, but FHS and FLS captured it well due to their flexibility. It was found that the proposed methods fitted successfully for smoothly transitioning functional time series data. In the data where abrupt changes did not exist, the curve could be estimated better by FLS.
    \item[Scenario 3:] The distinction between FHS and the other methods became evident. FHS was able to estimate the trend almost horizontally where the amount of change is zero and made large changes only at the discontinuous points. In contrast, FLS had less ability to make sparse estimates, resulting in a gently curved estimate. This suggests choosing an appropriate prior is essential.
    \item[Scenario 4:] BART and FHS results eclipse those of the other methods. In fact, FHS is inferior to BART but still guarantees the flexibility to capture abrupt changes.
\end{itemize}
As a result, the FHS has a basically favorable performance, and while FHS is occasionally inferior in performance, the difference was slight and the usefulness of FHS should not be compromised even in this situation.
With respect to CP, FHS always achieves a value close to $95$, indicating that it is a better inference than other Bayesian methods in terms of uncertainty evaluation.
Moreover, the stability of estimation accuracy and coverage of FHS compared to BART, a mere nonparametric method, indicates that it is more successful in analyzing it as functional time series data.


Next, we investigated how the credible intervals change with the sample size.
We changed the number $H$ of data at each time from 120 to 60 and compared that with the original settings with respect to MCIW and MC.
According to Table \ref{tab:h60}, the more data we got, the narrower the range of CIs was, even though the MC remained almost the same.
This result is consistent with the fact that uncertainty decreases with more data. This suggests that in trend estimation of functional data, not only the number of functions but also the number of observation points for each function is essential.

\begin{table}[h]
\caption{MCIW (mean credible interval width), MASVD (mean absolute sequentially variational deviation), MC (mean coverage) of FHS with horseshoe prior for $H=120$ and for $H=60$.\label{tab:h60}} 
\centering
\medskip
\begin{tabular}{cccccccccccc} 
\toprule
Scenario & H  &   & MCIW & MC \\ 
\midrule
1& 120  & & 3.321 & 0.981  \\ 
& 60    & & 4.879 & 0.985 \\ 
\midrule
2& 120  & & 5.214 & 0.966 \\ 
& 60    & & 6.513 & 0.964 \\ 
\midrule
3& 120  & & 4.350 & 0.982  \\ 
& 60    & & 6.234 & 0.979 \\ 
\midrule
4& 120  & & 3.988 & 0.937 \\ 
& 60    & & 5.389 & 0.940 \\ 
\bottomrule
\end{tabular}
\end{table}

In the literature dealing with similar models \citep[e.g.][]{yang2017efficient}, the number of basis is fixed. Also, in the classical approaches \citep[e.g.][]{martinez2021nonparametric,kowal2019functional}, ad hoc choices are made, such as considering cumulative explained variance. 
Those approaches took a larger number of basis than necessary.
Then we investigated whether the choice of basis is meaningful, i.e., the advantage of basis selection over preparing many basis functions. Hence, we fixed the number of basis functions at $25$ and implemented our method.
The results are shown in Table \ref{tab:basis}. It can be seen that, in general, the accuracy was better with the choice of bases.
Hence, choosing the number of bases played an essential role in increasing the accuracy. 
Since this model is a fully Bayesian framework, we could use a simple selection criterion, PPL. This is one advantage of the Bayesian models.

\begin{table*}[p]
\caption{MAD (mean absolute deviation), MCIW (mean credible interval width) and MC (mean coverage) for FHS (the estimator using the horseshoe-like prior) with different number of basis for each scenario.\label{tab:basis}} 
\centering
\medskip
\begin{tabular}{cclcccccccccc} 
\toprule
Scenario & & Method & $L$ &  & MAD & MCIW  & MC  \\
\midrule
1&& selected & 17.9 & & 0.538 & 3.321 &  0.981 \\  
&& fixed& 25.0 & & 0.571 & 3.658 &  0.986 \\  
\midrule
2&&  selected & 23.4 & & 0.961 & 5.214 &  0.966 \\  
&& fixed& 25.0&  & 0.981 & 5.362 &  0.967 \\  
\midrule
3&& selected& 21.5 & & 0.713 & 4.350 &  0.982  \\  
&& fixed& 25.0 &  & 0.749 & 4.572 & 0.982 \\  
\midrule
4&& selected & 19.0 & & 0.814 & 3.988 &  0.937 \\  
&&  fixed & 25.0 &  & 0.896 & 4.361 &  0.935  \\  
\bottomrule
\end{tabular}
\end{table*}

\begin{table*}[]
\caption{MAD (mean absolute deviation), MCIW (mean credible interval width) and MC (mean coverage) for FHS (the estimator based on the horseshoe-like prior) for each scenario under heterogeneously observed points.
\label{tab:miss}} 
\centering
\medskip
\begin{tabular}{cccccccccccc} 
\toprule
Scenario & omitted rate &  & MAD & MCIW &  MC\\ 
\midrule
1& 0$\%$& & 0.538 & 3.321 &  0.981 \\  
& 5$\%$ & & 0.674 & 4.335 &  0.985 \\  
& 10$\%$ && 0.727 &	4.622 &  0.985 \\  
\midrule
2& 0$\%$& & 0.961 & 5.214 &  0.966 \\ 
& 5$\%$&  & 1.017 & 5.599 &	 0.968 \\  
& 10$\%$& & 1.066 & 5.822 &  0.968 \\
\midrule
3& 0$\%$& & 0.713 & 4.350 &  0.982 \\  
& 5$\%$&  & 0.821 & 5.056 &  0.982 \\ 
& 10$\%$& & 0.862 & 5.277 &  0.980 \\   
\midrule
4& 0$\%$& & 0.814 & 3.988 &  0.937 \\  
& 5$\%$&  & 0.889 & 4.784 &  0.964 \\  
& 10$\%$& & 0.889 & 4.806 &  0.964 \\  
\bottomrule
\end{tabular}
\end{table*}

\subsection{Heterogeneously observed data}
We examined the cases where $5\%$ and $10\%$ of the whole data were omitted at random. 
We report the results in Table \ref{tab:miss} based on 3,000 MCMC iterations obtained after a burn-in period of 3,000 iterations. As the percentage of omitted data increased, the data became more unequally spaced and the estimation became less precise, but it was still able to capture the trend accurately.
The wide MCIW implied the increase in uncertainty due to the omitted data.
In addition, it should have been challenging to estimate where there is no data, but Figure \ref{fig:miss} suggests that the estimation could follow the trend including these points.
This provides evidence that heterogeneity in both the number of sampling points and the sampled locations does not make implementation challenging or have a significant negative impact on accuracy.

\begin{figure}[t]
  \begin{center}
  \includegraphics[width=\textwidth]{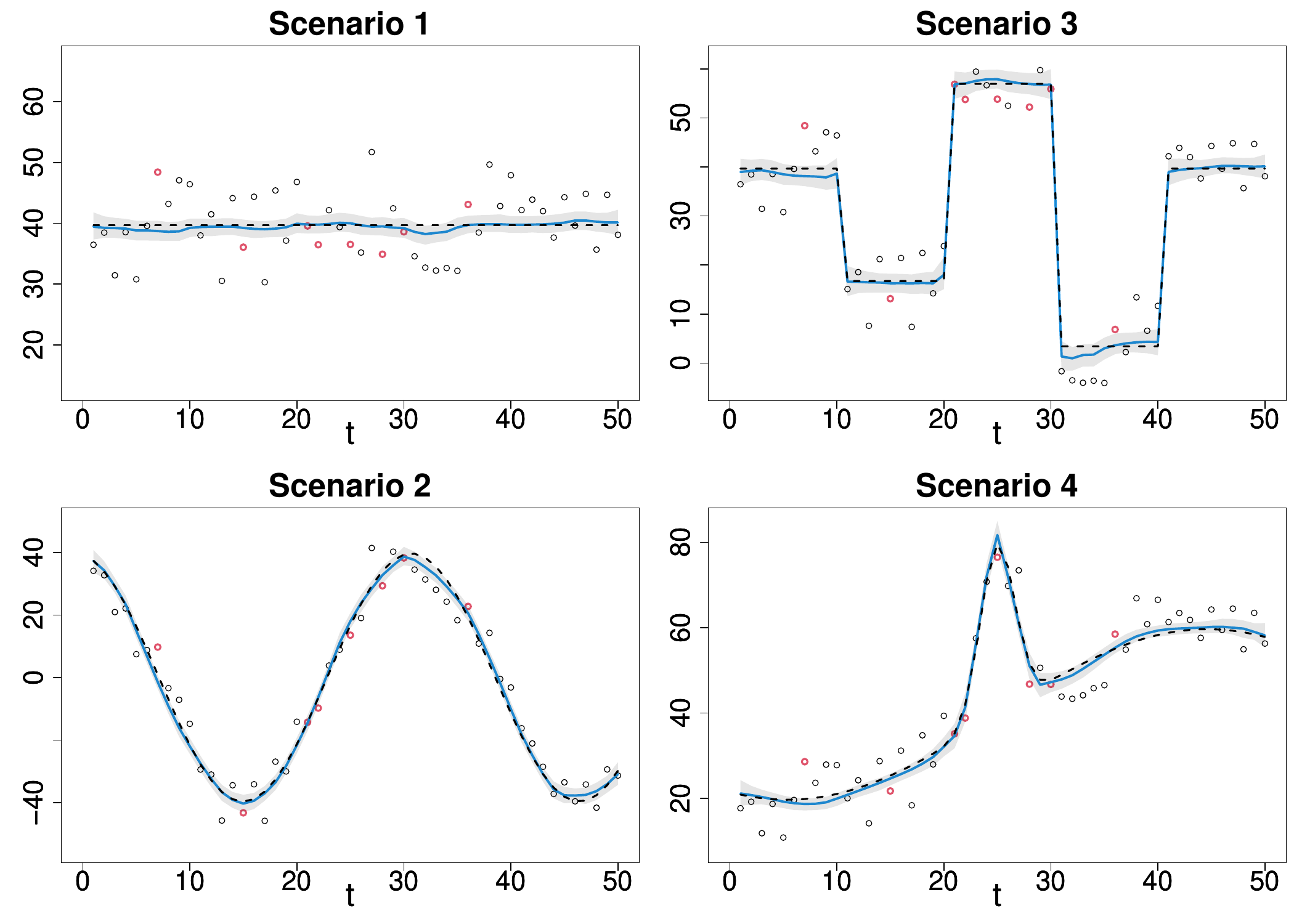}
  \end{center}
  \caption{The dotted line is the true trend, the blue line is the estimated trend, the gray area is the $95\%$ credible intervals at $x=40$, and red dots indicate omitted measurements. The data acquisition interval is uneven, and missing values exist.\label{fig:miss}}
\end{figure}

\section{Empirical Application}\label{sec:app}

Many studies have demonstrated the applicability and performance of functional time series analysis methods using age-specific fertility data \cite[e.g.][]{hyndman2007robust, wakayama2021locally}.
This section presents an empirical application of the proposed method using annual age-specific Australian fertility rates, which were obtained from the Australian Bureau of Statistics. These are defined as the number of births per 1,000 female residents.
This data covers the age group $15-49$ and the period $1921-2015.$
We then consider that there are $95$ functions with the domain $[15,49].$
Our interest in this application is the transition of the functions over time.

We applied functional horseshoe smoothing (FHS) and its Laplace prior version (FLS). 
The numbers of the basis for FHS and FLS were $27$ and $20$, chosen via PPL from $\{5,6,\ldots,30\}$. Also, the difference order $k$ was selected as 1 by PPL. The observation is shown in the upper left in Figure \ref{fig:asf:3d}, and its surface is rugged.
FLS smoothed the surface and largely removed the noise.
This is also the case for FHS, where the surface is smooth and the denoising effect is confirmed.
The difference is that FHS left the sharp edges intact, while FLS erased the sharp edges.
This implies that the FLS reduced the signal as well as the noise and hid change points.
Not doing it is the strength of FHS.

\begin{figure*}[p]
  \begin{center}
  \vspace*{-1cm}
  \includegraphics[width=13cm]{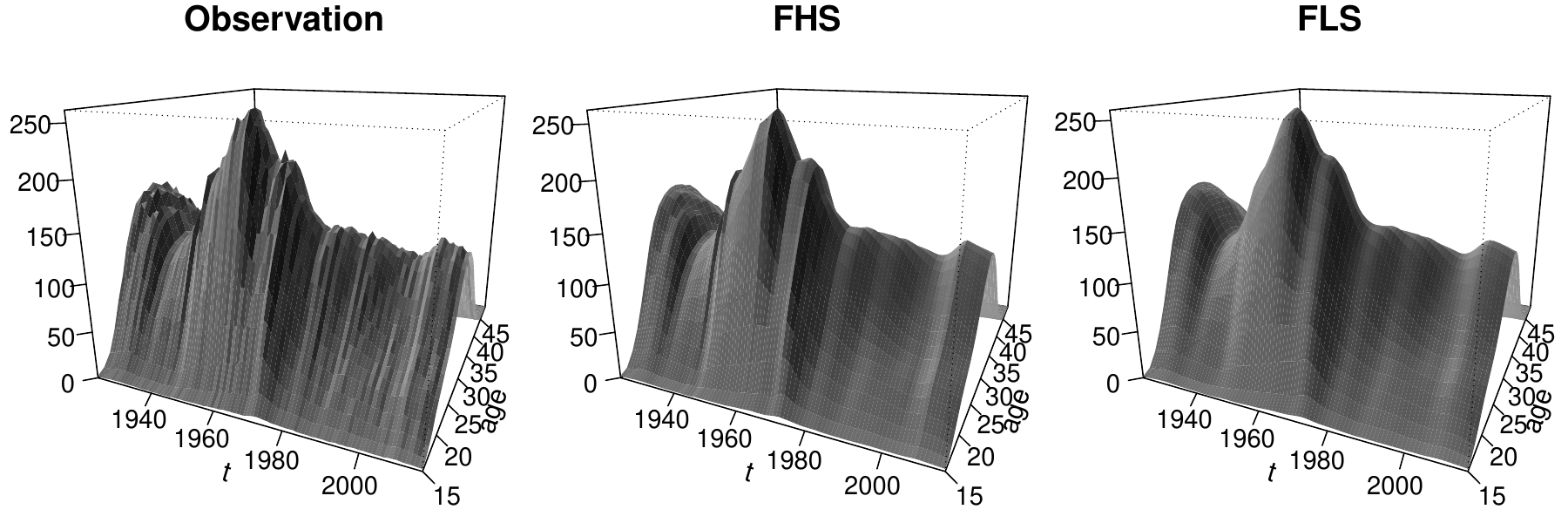}
  \end{center}
  \caption{The number of births per 1000 female residents by age in each year in Australia. The upper left is the observed quantity and the upper right (lower left) is the smoothed surface by FHS (FLS).\label{fig:asf:3d}}
\end{figure*}
\begin{figure*}[]
  \begin{center}
  \includegraphics[width=13cm]{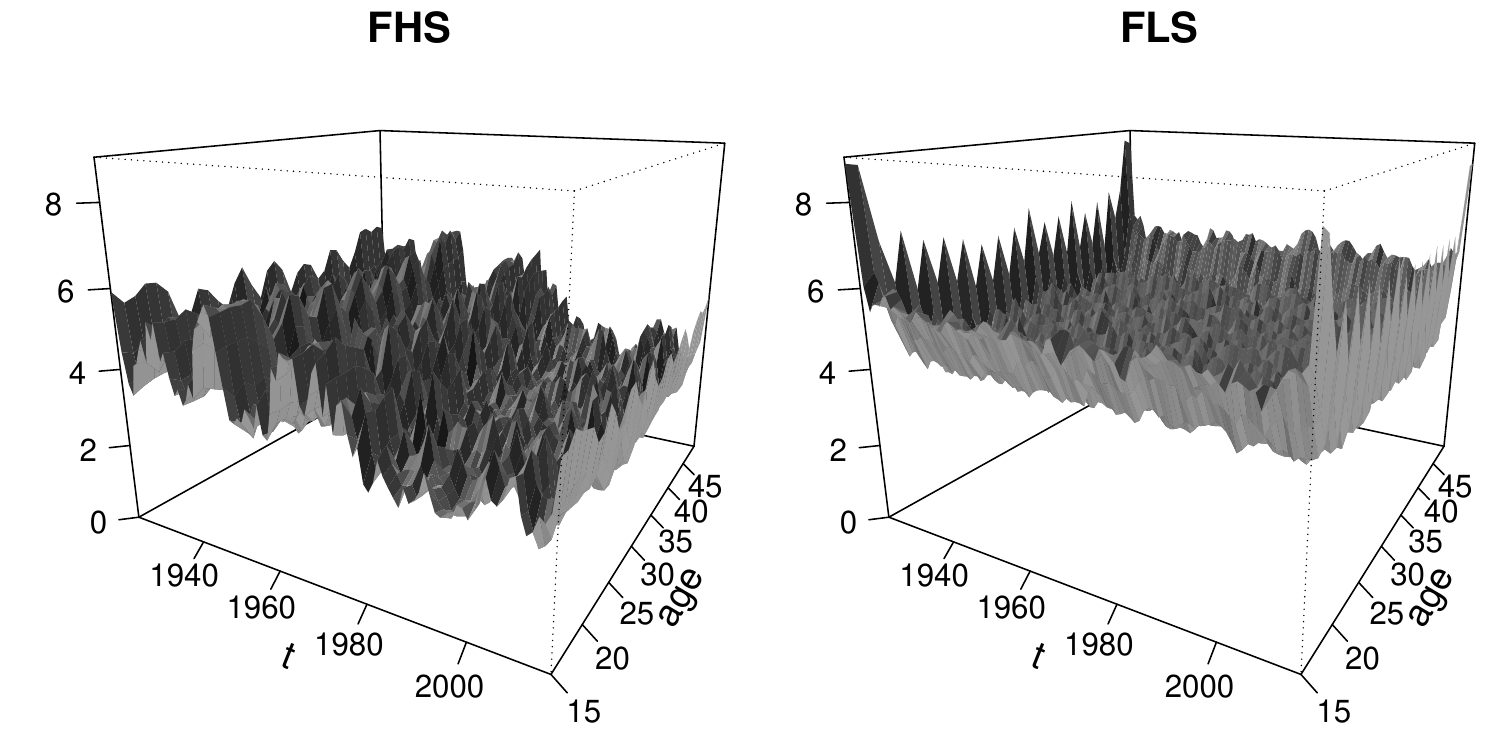}
  \end{center}
  \caption{The left (right) figure shows the credible interval width of FHS (FLS) in three dimensions.\label{fig:asf:ci}}
\end{figure*}

We next focus on the credible interval.
Figure \ref{fig:asf:ci} shows the difference between the 97.5 percentile and the 2.5 percentile for each year.
This is a three-dimensional representation of the size of the 95$\%$ credible region.
We see that FHS had a smaller credible area than FLS and regard FHS as a more plausible model.

\section{Discussion}
\label{sec:dis}
We presented the Bayesian nonparametric smoothing method for functional time series data. This enables locally adaptive estimation by exploiting the sparsity from the shrinkage prior distributions.
Simulation studies and empirical applications suggest that the method's performance, especially with a horseshoe-like prior, is good even with the presence of sharp changes.

Moreover, we elucidated its theoretical properties. We discussed two significant issues with the shrinkage prior distribution. The first is the spike at the origin of the marginal prior, and the second is the thickness of its tail.
Because of this, the estimator is expected to have ideal properties of eliminating noise and detecting abrupt changes simultaneously. 
We further checked the latter by proving the weak tail robustness of the posterior expectation. These are theoretical reason why horseshoe-like prior is favorable.

Functional horseshoe smoothing (FHS) has two advantages over functional trend filtering \citep{wakayama2021locally}. One is that selecting the parameters is easy. In the optimization approach, selecting tuning parameter (penalty parameter) using K-fold cross-validation was necessary, and the computational complexity increased with the number of folds and the number of parameters' candidates. 
The time and effort required to select three parameters, $k$ (the order of difference), $L$ (the number of basis functions), and $\lambda$ (the penalty parameter), is enormous.
In Bayesian approach, penalty parameters such as local parameters and global parameter are automatically selected through MCMC. The parameter selection is easier than with the frequentist method because it does not require the selection of penalty parameters and the PPL (criterion of model) can be calculated.
In addition, Bayesian models tended to perform better than frequentist methods, especially when flexible shrinkage is required. This is because, as \cite{polson2010shrink, carvalho2010horseshoe} noted, by using local and global parameters, the horseshoe prior can shrink each part of the estimation to a different degree. Hence, FHS allows the smooth and sharp parts of the estimation to coexist.

Also, FHS can deal with heterogeneously observed data. This model treats high-dimensional data as functional data, and as a result, allows heterogeneity in both the number of sampling points and the sampling locations.
Hence, this model does not require missing value completion and can estimate trends without being disturbed by heterogeneity.
In addition to being easy to implement, the results of simulation studies showed its good accuracy.

Our model is also useful for estimating VCFLM (varying-coefficient functional linear model).
VCFLM is a combination of a scalar-on-function model and a varying-coefficient model, where the coefficients are functions \cite{matsui2022truncated, wu2010varying, cardot2008varying} depending on exogenous variables. By expanding the predictor function and coefficient function in an orthonormal basis and introducing our prior into the coefficients of the basis expansion, the VCFLM can be analyzed in FHS framework.

\section*{Acknowledgments}
Research of the authors was supported in part by JSPS KAKENHI Grant Numbers 18H03628, 20H00080 and 21H00699 from Japan Society for the Promotion of Science.

\vspace{1cm}
\appendix
\begin{center}
{\large {\bf  Appendices}}
\end{center}

\section*{Appendix 1: Detailed forms of full conditional distributions of $\bm{b}_t$}
\label{append:post}
The full conditional distribution of $\bm{b}_t$ is $N(\bm{\mu}_t, \Sigma_t)$.
When $k=0$, it holds that
\begin{align*}
\bm{\mu}_t  = \left\{
    \begin{array}{ll}
         \Sigma_t\left(\lambda_t^{-2}\Pt^{\top}\Pt \bm{b}_{t+1}+\lambda_{t-1}^{-2}\Phi_{t-1}^{\top}\Phi_{t-1}\bm{b}_{t-1} +\tau^2  \Pt^{\top} \bm{y}_t \right)/(\tau^2\sigma^2)  & 2\leq t\leq T-1  \\
        \Sigma_t\left(\lambda_t^{-2}\Pt^{\top}\Pt \bm{b}_{t+1}+\tau^2  \Pt^{\top} \bm{y}_t \right)/(\tau^2\sigma^2) &   t=1 \\
        \Sigma_t\left(\lambda_{t-1}^{-2}\Phi_{t-1}^{\top}\Phi_{t-1}\bm{b}_{t-1} +\tau^2  \Pt^{\top} \bm{y}_t \right)/(\tau^2\sigma^2)    &  t=T \\
    \end{array} \right.
\end{align*}
and
\begin{align*}
    \Sigma_t^{-1}    =\left\{
    \begin{array}{ll}
        \left\{(\tau^2+\lambda_t^{-2})\Pt^{\top}\Pt + \lambda_{t-1}^{-2}\Phi_{t-1}^{\top}\Phi_{t-1}\right\} / (\tau^2\sigma^2) & 2\leq t\leq T-1  \\
        (\tau^2+\lambda_t^{-2})\Pt^{\top}\Pt  / (\tau^2\sigma^2) &   t=1 \\
        \tau^2\Pt^{\top}\Pt + \lambda_{t-1}^{-2}\Phi_{t-1}^{\top}\Phi_{t-1} / (\tau^2\sigma^2)   &  t=T.\\
    \end{array} \right.
\end{align*}
Also, when $k=1$, mean vector and precision matrix are \begin{align*}
   \bm{\mu}_t  = \left\{
    \begin{array}{ll}
        \Sigma_t\left(\frac{\Pt^{\top}\Pt(2\bm{b}_{t+1}-\bm{b}_{t+2})}{\sigma^2\tau^2\lambda^2_{t}}
        +\frac{\Pt^{\top} \bm{y}_t}{\sigma^2} \right) &  t=1 \\
        \Sigma_t\left(\frac{\Pt^{\top}\Pt(2\bm{b}_{t+1}-\bm{b}_{t+2})}{\sigma^2\tau^2\lambda^2_{t}}+\frac{2\Phi^{\top}_{t-1}\Phi_{t-1}(\bm{b}_{t+1}+\bm{b}_{t-1})}{\sigma^2\tau^2\lambda^2_{t-1}}
        +\frac{\Pt^{\top} \bm{y}_t}{\sigma^2} \right) &  t=2 \\
        \Sigma_t\left(\frac{\Pt^{\top}\Pt(2\bm{b}_{t+1}-\bm{b}_{t+2})}{\sigma^2\tau^2\lambda^2_{t}}+\frac{2\Phi^{\top}_{t-1}\Phi_{t-1}(\bm{b}_{t+1}+\bm{b}_{t-1})}{\sigma^2\tau^2\lambda^2_{t-1}}+\frac{\Phi_{t-2}^{\top}\Phi_{t-2}(2\bm{b}_{t-1}-\bm{b}_{t-2})}{\sigma^2\tau^2\lambda^2_{t-2}}
        +\frac{\Pt^{\top} \bm{y}_t}{\sigma^2} \right) & 3\leq t\leq T-2  \\
        \Sigma_t\left(\frac{2\Phi^{\top}_{t-1}\Phi_{t-1}(\bm{b}_{t+1}+\bm{b}_{t-1})}{\sigma^2\tau^2\lambda^2_{t-1}}+\frac{\Phi_{t-2}^{\top}\Phi_{t-2}(2\bm{b}_{t-1}-\bm{b}_{t-2})}{\sigma^2\tau^2\lambda^2_{t-2}}
        +\frac{\Pt^{\top} \bm{y}_t}{\sigma^2} \right) &  t=T-1 \\
        \Sigma_t\left(\frac{\Phi_{t-2}^{\top}\Phi_{t-2}(2\bm{b}_{t-1}-\bm{b}_{t-2})}{\sigma^2\tau^2\lambda^2_{t-2}}
        +\frac{\Pt^{\top} \bm{y}_t}{\sigma^2} \right) &  t=T \\
    \end{array} \right.
\end{align*}
and
\begin{align*}
    \Sigma_t^{-1}    =\left\{
    \begin{array}{ll}
        \frac{\tau^2\lambda_t^2+1}{\sigma^2\tau^2\lambda_t^2}\Phi_t^{\top}\Phi_t &  t= 1  \\
        \frac{\tau^2\lambda_t^2+1}{\sigma^2\tau^2\lambda_t^2}\Phi_t^{\top}\Phi_t+\frac{4}{\sigma^2\tau^2\lambda_{t-1}^2}\Phi^{\top}_{t-1}\Phi_{t-1} &  t= 2  \\
        \frac{\tau^2\lambda_t^2+1}{\sigma^2\tau^2\lambda_t^2}\Phi_t^{\top}\Phi_t+\frac{4}{\sigma^2\tau^2\lambda_{t-1}^2}\Phi^{\top}_{t-1}\Phi_{t-1}+\frac{1}{\sigma^2\tau^2\lambda_{t-2}^2}\Phi^{\top}_{t-2}\Phi_{t-2} & 3\leq t\leq T-2  \\
        \frac{\tau^2\lambda_t^2}{\sigma^2\tau^2\lambda_t^2}\Phi_t^{\top}\Phi_t+\frac{4}{\sigma^2\tau^2\lambda_{t-1}^2}\Phi^{\top}_{t-1}\Phi_{t-1}+\frac{1}{\sigma^2\tau^2\lambda_{t-2}^2}\Phi^{\top}_{t-2}\Phi_{t-2} &  t= T-1  \\
        \frac{\tau^2\lambda_t^2}{\sigma^2\tau^2\lambda_t^2}\Phi_t^{\top}\Phi_t+\frac{1}{\sigma^2\tau^2\lambda_{t-2}^2}\Phi^{\top}_{t-2}\Phi_{t-2} &  t= T
    \end{array} \right.
\end{align*}

Assume data is observed homogeneously, that is, $\Phi_1=\Phi_2=\cdots=\Phi_T$.
When $k=0$, mean vector and precision matrix can be written as 
\begin{align*}
\bm{\mu}_t  = \left\{
    \begin{array}{ll}
        \left\{ c_t(\lambda_t^{-2} \bm{b}_{t+1}+\lambda_{t-1}^{-2}\bm{b}_{t-1}) +\tau^2 \Sigma_t \Phi^{\top} \bm{y}_t \right\}/(\tau^2\sigma^2)  & 2\leq t\leq T-1  \\
        \left( c_t \lambda_t^{-2} \bm{b}_{t+1} +\tau^2 \Sigma_t \Phi^{\top} \bm{y}_t \right)/(\tau^2\sigma^2) &   t=1 \\
        \left(c_t\lambda_{t-1}^{-2}\bm{b}_{t-1}) +\tau^2 \Sigma_t \Phi^{\top} \bm{y}_t \right)/(\tau^2\sigma^2)   &  t=T \\
    \end{array} \right.
\end{align*}
and $\Sigma_t = c_t(\Phi^\top\Phi)^{-1}$ with
\begin{align*}
    c_t   =\left\{
    \begin{array}{ll}
        \tau^2\sigma^2/(\tau^2+\lambda_t^{-2}+\lambda_{t-1}^{-2}) & 2\leq t\leq T-1  \\
        \tau^2\sigma^2/(\tau^2+\lambda_t^{-2}) &   t=1 \\
        \tau^2\sigma^2/(\tau^2+\lambda_{t-1}^{-2})   &  t=T.\\
    \end{array} \right.
\end{align*}
Furthermore, when $k=1$, the forms of $\bm{\mu}_t$ and $c_t$ are 
\begin{equation*}
\bm{\mu}_t = 
    \begin{cases}
        \begin{split}
        & [c_t\{ (2\lambda_t^{-2}+2\lambda_{t-1}^{-2}) \bm{b}_{t+1}+ (2\lambda_{t-1}^{-2} +2\lambda_{t-2}^{-2})\bm{b}_{t-1}
        -\lambda_{t}^{-2}\bm{b}_{t+2}-\lambda_{t-2}^{-2}\bm{b}_{t-2} \}   \\
        & \qquad \qquad \qquad +\tau^2 \Sigma_t \Phi^{\top} \bm{y}_t ]/(\tau^2\sigma^2)
        \end{split} &  3\leq t\leq T-2 \\
        \left[ c_t(2\lambda_t^{-2}\bm{b}_{t+1}-\lambda_{t}^{-2}\bm{b}_{t+2} ) +\tau^2 \Sigma_t \Phi^{\top} \bm{y}_t \right]/(\tau^2\sigma^2) &   t=1\\
        \left[ c_t\{ (2\lambda_t^{-2}+2\lambda_{t-1}^{-2}) \bm{b}_{t+1}+ 2\lambda_{t-1}^{-2}\bm{b}_{t-1}
        -\lambda_{t}^{-2}\bm{b}_{t+2} \} +\tau^2 \Sigma_t \Phi^{\top} \bm{y}_t \right]/(\tau^2\sigma^2)   &   t=2 \\
        \left[ c_t\{ 2\lambda_{t-1}^{-2} \bm{b}_{t+1}+ (2\lambda_{t-1}^{-2} +2\lambda_{t-2}^{-2})\bm{b}_{t-1}
        -\lambda_{t-2}^{-2}\bm{b}_{t-2} \} +\tau^2 \Sigma_t \Phi^{\top} \bm{y}_t \right]/(\tau^2\sigma^2)    &  t=T-1 \\
        \left[ c_t( 2\lambda_{t-2}^{-2}\bm{b}_{t-1}
        -\lambda_{t-2}^{-2}\bm{b}_{t-2} ) +\tau^2 \Sigma_t \Phi^{\top} \bm{y}_t \right]/(\tau^2\sigma^2)    &  t=T \\
    \end{cases}
\end{equation*}
and 
\begin{align*}
    c_t   =\left\{
    \begin{array}{ll}
        \tau^2\sigma^2/(\tau^2+\lambda_t^{-2}+4\lambda_{t-1}^{-2}+\lambda_{t-2}^{-2}) & 3\leq t\leq T-2  \\
        \tau^2\sigma^2/(\tau^2+\lambda_t^{-2}+4\lambda_{t-1}^{-2}) &   t=2 \\
        \tau^2\sigma^2/(\tau^2+\lambda_t^{-2}) &   t=1 \\
        \tau^2\sigma^2/(\tau^2+4\lambda_{t-1}^{-2}+\lambda_{t-2}^{-2})&  t=T-1\\
        \tau^2\sigma^2/(\tau^2+\lambda_{t-2}^{-2})&  t=T.\\
    \end{array} \right.
\end{align*}

\section*{Appendix 2: Proof of the propositions}\label{append:proof}
This section provides proofs of Proposition \ref{prop:pri1} and Proposition \ref{prop:tail}.

\begin{proof}[Proof of Proposition \ref{prop:pri1} $(i)$]
For $L=1$, the divergence of $\pi(\bm{\delta}_t \mid\tau,\sigma)$ follows from the proof of Theorem 1 in \cite{carvalho2010horseshoe}.
We restrict our investigation to $L\ge2$.
For simplicity, we fix $\sigma=\tau=1$. We have
\begin{align*}
    \pi(\bm{\delta}_t ) &\propto
\int_{0}^{\infty}\frac{1}{\lambda^L_t(1+\lambda_t^2)}\exp\left\{ -\frac{1}{2\lambda_t^2}\bm{\delta}^{\top}_t \Pt^{\top} \Pt \bm{\delta}_t \right\}d\lambda_t\\
&\ge \int_{0}^{1}\frac{1}{\lambda^L_t(1+\lambda_t^2)}\exp\left\{ -\frac{1}{2\lambda_t^2}\bm{\delta}^{\top}_t \Pt^{\top} \Pt \bm{\delta}_t \right\}d\lambda_t.
\end{align*}
Let $u=\lambda_t^{-2}$. Then, we obtain
\begin{align*}
    \pi(\bm{\delta}_t ) 
&\ge 2\int_{1}^{\infty}\frac{u^{\frac{L-1}{2}}}{1+u}\exp\left\{-\frac{u}{2}\bm{\delta}^{\top}_t \Pt^{\top} \Pt \bm{\delta}_t\right\}du\\
&\ge \int_{1}^{\infty}u^{\frac{L-3}{2}}\exp\left\{-\frac{u}{2}\bm{\delta}^{\top}_t \Pt^{\top} \Pt \bm{\delta}_t\right\}du.
\end{align*}
Also, we transform $\tilde{u}=\frac{\bm{\delta}^{\top}_t \Pt^{\top} \Pt \bm{\delta}_t}{2} u$. It leads to 
\begin{align*}
    \pi(\bm{\delta}_t)
\ge\left(\frac{\sqrt{2}}{\|\Phi\bm{\delta_t}\|_2}\right)^{L-1}
\int_{\frac{\|\Phi\bm{\delta_t}\|_2^2}{2}}^{\infty}\tilde{u}^{\frac{L-3}{2}}\exp\left\{-\tilde{u}\right\}d\tilde{u}.
\end{align*}The right-hand equals to\begin{align*}
    \left(\frac{\sqrt{2}}{\|\Phi\bm{\delta_t}\|_2}\right)^{L-1}\Gamma\left(\frac{L-1}{2}, \frac{\|\Phi\bm{\delta_t}\|_2^2}{2}\right),
\end{align*}
where $\Gamma(\cdot,\cdot)$ is the upper incomplete gamma function. If $\bm{\delta}_t\rightarrow 0$, it holds that
\begin{align*}
    \frac{\sqrt{2}}{\|\Phi\bm{\delta_t}\|_2} \rightarrow \infty\ ,\ \ \ \ \ 
    \Gamma\left(\frac{L-1}{2},\frac{\|\Phi\bm{\delta_t}\|_2^2}{2}\right)\rightarrow \Gamma\left(\frac{L-1}{2}\right)
\end{align*} Hence, $\pi(\bm{\delta}_t) \rightarrow\infty$ as $\bm{\delta}_t\rightarrow 0$.
\end{proof}

Before proving Proposition \ref{prop:pri1} $(ii)$, we introduce an important result.
First, slight modification with Lemma S5 in \cite{hamura2020log} leads to the following lemma, which gives rise to Corollary \ref{lemma:HIS}.
\begin{lemma}\label{lemma:HIS}
Let $f(\cdot)$ and $g(\cdot)$ be positive and continuous functions such that $e^{-\frac{1}{x}}f(x)$ and $e^{-\frac{1}{x}}g(x)$ are integrable on the real line. Then, it holds that\begin{align}\label{lim}    \frac{ \lim_{x\rightarrow\infty} \int_0^{\infty}\phi(x; 0, v) f(v) dv}{\lim_{x\rightarrow\infty} \int_0^{\infty}\phi(x; 0, v) g(v) dv} = \lim_{v\rightarrow\infty} \frac{f(v)}{g(v)},
\end{align} where $\phi(x; 0, v)$ is p.d.f. of $N(0,v)$.\end{lemma}

\begin{corollary}\label{cor:his}
If $\lim_{x\to\infty} f(x)/g(x)=1$, we denote $f\approx g$.
Let $f(v)=\frac{1}{v^{L/2}(1+v)}$ and $g(v)=\frac{1}{v^{L/2+1}}$ for fixed $L\in \mathbb{N}$. Then, $\int_0^{\infty}\phi(x; 0, v) f(v) dv\approx \int_0^{\infty}\phi(x; 0, v) g(v) dv$ holds. 
\end{corollary}
\begin{proof}[Proof of Corollary \ref{cor:his}]
These are obviously positive and continuous. Also $e^{-\frac{1}{x}}f(x)$ and $e^{-\frac{1}{x}}g(x)$ are integrable on the real line, since \begin{align*}
 \int_0^{\infty}f(v)e^{-\frac{1}{v}}dv 
 &\leq \int_0^{\infty}\frac{1}{v^{\frac{L}{2}+1}}e^{-\frac{1}{v}} dv+ 
 \int_0^{\infty}\frac{1}{v^{\frac{L}{2}}}e^{-\frac{1}{v}}dv\\
 &\leq \Gamma (\frac{L+1}{2})+\Gamma (\frac{L-1}{2})\\
 &<\infty,
\end{align*}and similarly $\int_0^{\infty}g(v)\phi(x; 0, v)dv<\infty.$ Then, Lemma \ref{lemma:HIS} and $\lim_{v\rightarrow\infty}f(v)/g(v)=1$ lead to the result.
\end{proof}

Based on the above results, we prove Proposition \ref{prop:pri1} $(ii)$ as follows.
\begin{proof}[Proof of Proposition \ref{prop:pri1} $(ii)$]
For simplicity, we fix $\sigma=\tau=1$.
With transformation $\bar{u}= \lambda_t^2$,
\begin{align*}
\pi(\bm{\delta}_t \mid \tau,\sigma) &\propto
\int_{0}^{\infty}\frac{1}{\lambda^L_t(1+\lambda_t^2)}\exp\left\{ -\frac{1}{2\sigma^2\tau^2\lambda_t^2}\bm{\delta}^{\top}_t \Pt^{\top} \Pt \bm{\delta}_t \right\}d\lambda_t\\
& = \frac12 \int_{0}^{\infty} \frac{1}{\bar{u}^{\frac{L}{2}}(1+\bar{u})} \frac{1}{\bar{u}^{\frac12}}\exp \left(-\frac{z^2}{2\bar{u}}\right)d\bar{u},
\end{align*}
where $z = \|\Pt \bm{\delta}_t\|_2$
From Corollary \ref{lemma:HIS}, \begin{align*}
    \frac12 \int_{0}^{\infty} \frac{1}{\bar{u}^{\frac{L}{2}}(1+\bar{u})} \frac{1}{\bar{u}^{\frac12}}\exp \left(-\frac{z^2}{2\bar{u}}\right)d\bar{u} 
    & \approx \frac12 \int_{0}^{\infty} \frac{1}{\bar{u}^{\frac{L}{2}+1}} \frac{1}{\bar{u}^{\frac12}}\exp \left(-\frac{z^2}{2\bar{u}}\right)d\bar{u}\\
    & = \frac12\frac{\Gamma(\frac{L+1}{2})}{(\frac{z^2}{2})^{\frac{L+1}{2}}}\\
    & = O(z^{-L-1}).
\end{align*}
\end{proof}

Next, we prove the weak tail robustness of the posterior mean.
\begin{proof}[Proof of Proposition \ref{prop:tail}]

Also, we define \begin{align*}
    \bm{z} &\equiv (\bm{z}_1^{\top},...,\bm{z}_{T-1}^{\top} )^{\top} \sim N(\bm{\eta}, \Sigma_{z}),\\
    \bm{\eta} &\equiv (\Phi\bm{\delta}_1^{\top},...,\Phi\bm{\delta}_{T-1}^{\top} )^{\top} \sim \prod_{t=1}^{T-1}N(\bm{0},\lambda_t^2 \Sigma_{\delta}),\\
    \lambda_t^2 &\sim \prod_{t=1}^{T-1} \pi(\lambda_t) = \prod_{t=1}^{T-1}\frac{1}{1+\lambda_t^2},
\end{align*}where $\Sigma_z \equiv 2\sigma^2I_n\otimes I_L$ and $\Sigma_{\delta}\equiv\sigma^2\tau^2(\Pt^\top \Pt)^{-1}$.

Then we rewrite the posterior mean by score function. When we denote \begin{align*}
    m(\bm{z}) &= \int N(\bm{z} \mid \bm{\eta}, \Sigma_{z})p(\bm{\eta}) d\bm{\eta}, \\
    \frac{\partial m(\bm{z})}{\partial\bm{z}} &= \int - \Sigma_z(\bm{z}-\bm{\eta})N(\bm{z} \mid \bm{\eta}, \Sigma_{z})p(\bm{\eta}) d\bm{\eta},
\end{align*}
the following representation is available,\begin{align*}
   {\rm E}[\bm{\eta}|\bm{z}] - \bm{z} = \Sigma_{z}\frac{1}{m(\bm{z})}\frac{\partial m(\bm{z})}{\partial\bm{z}}.
\end{align*}
The goal is to find how this behaves as $\bm{z}_t^*\to \infty.$ Hence, we analyze $m(\bm{z})$ first.

\textbf{Step 1: Analysis of} $m(\bm{z})$. With $T= \mathrm{blockdiag} (\lambda_1\Sigma_{\delta}^{1/2},...,\lambda_{T-1}\Sigma_{\delta}^{1/2})$, \begin{align*}
    m(\bm{z})&\propto \int\int    
    \left(\prod_{t=1}^{T-1}\frac{\pi(\lambda_t)}{\lambda_t^L}\right)\exp\left\{ -\frac{1}{2}(\bm{\eta}-\bm{z})^{\top}\Sigma_z^{-1}(\bm{\eta}-\bm{z}) -\frac12 \bm{\eta}^{\top}T^{-2}\bm{\eta} \right\}
    d\bm{\eta}d\bm{\lambda}\\
    & = \int\int    
    \exp\left\{ -\frac{1}{2}\{\bm{\eta}- (\Sigma_z^{-1}+T^{-2})^{-1}\Sigma_z^{-1}\bm{z}  \}^{\top}(\Sigma_z^{-1}+T^{-2})\{\bm{\eta}- (\Sigma_z^{-1}+T^{-2})^{-1}\Sigma_z^{-1}\bm{z}  \} \right\}
    d\bm{\eta}\\
    &\qquad \qquad \times \left(\prod_{t=1}^{T-1}\frac{\pi(\lambda_t)}{\lambda_t^L}\right) \exp\left\{ -\frac{1}{2}\bm{z}^{\top} \Sigma_z^{-1} \{ \Sigma_z -(\Sigma_z^{-1}+T^{-2})^{-1}\} \Sigma_z^{-1} \bm{z}\right\} d\bm{\lambda}\\
    & \propto \int\int    
    N\left(\bm{\eta}\mid (\Sigma_z^{-1}+T^{-2})^{-1}\Sigma_z^{-1}\bm{z}, (\Sigma_z^{-1}+T^{-2})^{-1} \right)
    d\bm{\eta}\\
    &\qquad \qquad \times \frac{1}{|\Sigma_z^{-1}+T^{-2}|^{1/2}}
    \left(\prod_{t=1}^{T-1}\frac{\pi(\lambda_t)}{\lambda_t^L}\right) \exp\left\{ -\frac{1}{2}\bm{z}^{\top} \Sigma_z^{-1} \{ \Sigma_z -(\Sigma_z^{-1}+T^{-2})^{-1}\} \Sigma_z^{-1} \bm{z}\right\} d\bm{\lambda}\\
    &= \int \frac{1}{|\Sigma_z^{-1}+T^{-2}|^{1/2}}
    \left(\prod_{t=1}^{T-1}\frac{\pi(\lambda_t)}{\lambda_t^L}\right) \exp\left\{ -\frac{1}{2}\bm{z}^{\top} \Sigma_z^{-1} \{ \Sigma_z -(\Sigma_z^{-1}+T^{-2})^{-1}\} \Sigma_z^{-1} \bm{z}\right\} d\bm{\lambda}
\end{align*}

From Woodbury matrix identity, we see \begin{align*}
    m(\bm{z}) &\propto \int\frac{\prod_{t=1}^{T-1}\lambda_t^L |\Sigma_{\delta}|}{|I + T^{-1}\Sigma_z^{-1}T^{-1}|^{1/2}}
    \left(\prod_{t=1}^{T-1}\frac{\pi(\lambda_t)}{\lambda_t^L}\right)  \exp\left\{ -\frac{1}{2}\bm{z}^{\top} T^{-1} \{ I+T^{-1}\Sigma_z T^{-1}\}^{-1} T^{-1} \bm{z}\right\} d\bm{\lambda}\\
    &\propto \int \frac{\prod_{t=1}^{T-1}\pi(\lambda_t)}{|I + T^{-1}\Sigma_z^{-1}T^{-1}|^{1/2}}
    \exp\left\{ -\frac{1}{2}\bm{z}^{\top} T^{-1} \{ I+T^{-1}\Sigma_z T^{-1}\}^{-1} T^{-1} \bm{z}\right\} d\bm{\lambda}.
\end{align*}
Hence, we obtain \begin{align}
    \frac{1}{m(\bm{z})}\frac{\partial m(\bm{z})}{\partial\bm{z}} 
    &= \int_{\mathbb{R}_+^{T-1}} \frac{\prod_{t=1}^{T-1}\pi(\lambda_t)}{|I + T^{-1}\Sigma_z^{-1}T^{-1}|^{1/2}}\nonumber
    \exp\left\{ -\frac{1}{2}\bm{z}^{\top} T^{-1} \{ I+T^{-1}\Sigma_z T^{-1}\}^{-1} T^{-1} \bm{z}\right\}\\ \nonumber &\qquad \left\{ -T^{-1} \{ I+T^{-1}\Sigma_z T^{-1}\}^{-1} T^{-1} \bm{z}\right\}d\bm{\lambda}\\ 
    & \qquad / \int_{\mathbb{R}_+^{T-1}}\frac{\prod_{t=1}^{T-1}\pi(\lambda_t)}{|I + T^{-1}\Sigma_z^{-1}T^{-1}|^{1/2}}
    \exp\left\{ -\frac{1}{2}\bm{z}^{\top} T^{-1} \{ I+T^{-1}\Sigma_z T^{-1}\}^{-1} T^{-1} \bm{z}\right\} d\bm{\lambda}\label{eq:st1}
\end{align}

Define $v_t\equiv \frac{\lambda_t}{\|\bm{z}_t\|_2^2}$, $W \equiv \diag(\|\bm{z}_1\|_2^2,\ldots,\|\bm{z}_{T-1}\|_2^2 )$ and $\bm{\omega} \equiv (W^{-1}\otimes I_L)\bm{z}$. With $V\equiv\diag(v_1,\ldots,v_{T-1})$, $T = (VW) \otimes \Sigma_{\delta}^{1/2}$ holds.

Using these new variables, we calculate (\ref{eq:st1}) as \begin{align}
    & \int_{\mathbb{R}_+^{T-1}} \frac{\prod_{t=1}^{T-1}\pi(\|\bm{z}_t\|_2^2 v_t )}{|I + \{(VW) \otimes \Sigma_{\delta}^{1/2}\}\Sigma_z^{-1}\{(VW) \otimes \Sigma_{\delta}^{1/2}\}|^{1/2}} \nonumber \\
    &\qquad \exp\left\{ -\frac{1}{2}\bm{\omega}^{\top} \{V \otimes \Sigma_{\delta}^{1/2}\}^{-1} \{ I+\{(VW) \otimes \Sigma_{\delta}^{1/2}\}^{-1}\Sigma_z \{(VW) \otimes \Sigma_{\delta}^{1/2}\}^{-1}\}^{-1} \{V \otimes \Sigma_{\delta}^{1/2}\}^{-1} \bm{\omega}\right\}   \nonumber  \\ 
    &\qquad \left\{ -\{(VW) \otimes \Sigma_{\delta}^{1/2}\}^{-1} \{ I+\{(VW) \otimes \Sigma_{\delta}^{1/2}\}^{-1}\Sigma_z \{(VW) \otimes \Sigma_{\delta}^{1/2}\}^{-1}\}^{-1} \{V \otimes \Sigma_{\delta}^{1/2}\}^{-1} \bm{\omega}\right\}d\bm{v}  \nonumber  \\
    & \quad / \int_{\mathbb{R}_+^{T-1}}\frac{\prod_{t=1}^{T-1}\pi(\|\bm{z}_t\|_2^2 v_t )}{|I + \{(VW) \otimes \Sigma_{\delta}^{1/2}\}\Sigma_z^{-1}\{(VW) \otimes \Sigma_{\delta}^{1/2}\}|^{1/2}} \nonumber \\
    &\qquad \exp\left\{ -\frac{1}{2}\bm{\omega}^{\top} \{V \otimes \Sigma_{\delta}^{1/2}\}^{-1} \{ I+\{(VW) \otimes \Sigma_{\delta}^{1/2}\}^{-1}\Sigma_z \{(VW) \otimes \Sigma_{\delta}^{1/2}\}^{-1}\}^{-1} \{V \otimes \Sigma_{\delta}^{1/2}\}^{-1} \bm{\omega}\right\} d\bm{v}.\label{eq:st2}
\end{align}

\textbf{Step 2: Limit of the integrand.} 
To evaluate (\ref{eq:st2}) by the dominated convergence theorem, we need the integrable dominating function and limit of the integrand. We calculate the limits in this step and find the dominating function in the next step.
Let arbitrarily chosen index be $t^*\in \{1,...,T-1\}$ and fix $\bm{z}_t$ for $t\neq t^*.$ We observe (\ref{eq:st2}) when $\|\bm{z}_{t^*}\|_2\to \infty.$

We provide the following four approximations. 
\begin{itemize}
    \item[E1:] $\prod_{t=1}^{T-1}\pi(\|\bm{z}_t\|_2^2 v_t ) \sim  \prod_{t\neq t^*}\pi(\|\bm{z}_t\|_2^2 v_t )\pi(\|\bm{z}_{t^*}\|_2^2)/ v_{t^*}^2$
    \item[E2:] $|I + \{(VW) \otimes \Sigma_{\delta}^{1/2}\}\Sigma_z^{-1}\{(VW) \otimes \Sigma_{\delta}^{1/2}\}| \sim
    |W\otimes I_L| | \tilde{\Sigma}+V^2\otimes \Sigma_{\delta}| /|\Sigma_z|$
    \item[E3:] $\exp\left\{ -\frac{1}{2}\bm{\omega}^{\top} \{V \otimes \Sigma_{\delta}^{1/2}\}^{-1} \{ I+\{(VW) \otimes \Sigma_{\delta}^{1/2}\}^{-1}\Sigma_z \{(VW) \otimes \Sigma_{\delta}^{1/2}\}^{-1}\}^{-1} \{V \otimes \Sigma_{\delta}^{1/2}\}^{-1} \bm{\omega}\right\}\sim 
    \exp(-\frac12 \bm{\omega}^{\top} (V^2\otimes\Sigma_{\delta} +\tilde{\Sigma} )^{-1}\bm{\omega})$
    \item[E4:] $\{(VW) \otimes \Sigma_{\delta}^{1/2}\}^{-1} \{ I+\{(VW) \otimes \Sigma_{\delta}^{1/2}\}^{-1}\Sigma_z \{(VW) \otimes \Sigma_{\delta}^{1/2}\}^{-1}\}^{-1} \{V \otimes \Sigma_{\delta}^{1/2}\}^{-1} \bm{\omega} 
    \sim  A(v_{-i}) (I+A(v_{-i})\Sigma A(v_{-i}))^{-1} \{V \otimes \Sigma_{\delta}^{1/2}\}^{-1} \bm{\omega}$,
\end{itemize}
where 
\begin{align*}
    A(v_{-i}) &\equiv \lim_{\|\bm{z}_{t^*}\|_2\to \infty} \{(VW) \otimes \Sigma_{\delta}^{1/2}\}^{-1},\\    
    \tilde{\Sigma} &\equiv \lim_{\|\bm{z}_{t^*}\|_2\to \infty} (W\otimes I_L)^{-1} \Sigma (W\otimes I_L)^{-1}.
\end{align*}
E1, E3 and E4 are obvious, and E2 is immediately obtained from the following relationship \begin{align*}
    &|-\Sigma_z| |I - \{(VW) \otimes \Sigma_{\delta}^{1/2}\} (-\Sigma_z^{-1})\{(VW) \otimes \Sigma_{\delta}^{1/2}\}| \\
    &= \begin{vmatrix}
       I & (VW) \otimes \Sigma_{\delta}^{1/2} \\
       (VW) \otimes \Sigma_{\delta}^{1/2} & -\Sigma \\
      \end{vmatrix}\\
    &=|I| |-\Sigma_z - \{(VW) \otimes \Sigma_{\delta}^{1/2}\} I\{(VW) \otimes \Sigma_{\delta}^{1/2}\}|.
\end{align*}
Then, the limit of the integrand of (\ref{eq:st2}) is a combination of E1~E4.

\textbf{Step 3: Dominating function of the integrand.} Assume $\|\bm{z}_{t^*}\|_2 >> 1.$ We continue (\ref{eq:st2}) as 
\begin{align*}
    & \int_{\mathbb{R}_+^{T-1}} \frac{\prod_{t=1}^{T-1}\pi(\|\bm{z}_t\|_2^2 v_t )}{|I + \{(VW) \otimes \Sigma_{\delta}^{1/2}\}\Sigma_z^{-1}\{(VW) \otimes \Sigma_{\delta}^{1/2}\}|^{1/2}} \nonumber \\
    &\qquad \exp\left\{ -\frac{1}{2}\bm{\omega}^{\top} \{V \otimes \Sigma_{\delta}^{1/2}\}^{-1} \{ I+\{(VW) \otimes \Sigma_{\delta}^{1/2}\}^{-1}\Sigma_z \{(VW) \otimes \Sigma_{\delta}^{1/2}\}^{-1}\}^{-1} \{V \otimes \Sigma_{\delta}^{1/2}\}^{-1} \bm{\omega}\right\}   \nonumber  \\ 
    &\qquad \left\{ -\{(VW) \otimes \Sigma_{\delta}^{1/2}\}^{-1} \{ I+\{(VW) \otimes \Sigma_{\delta}^{1/2}\}^{-1}\Sigma_z \{(VW) \otimes \Sigma_{\delta}^{1/2}\}^{-1}\}^{-1} \{V \otimes \Sigma_{\delta}^{1/2}\}^{-1} \bm{\omega}\right\}d\bm{v}  \nonumber  \\
    & \quad / \int_{\mathbb{R}_+^{T-1}}\frac{\prod_{t=1}^{T-1}\pi(\|\bm{z}_t\|_2^2 v_t )}{|I + \{(VW) \otimes \Sigma_{\delta}^{1/2}\}\Sigma_z^{-1}\{(VW) \otimes \Sigma_{\delta}^{1/2}\}|^{1/2}} \nonumber \\
    &\qquad \exp\left\{ -\frac{1}{2}\bm{\omega}^{\top} \{V \otimes \Sigma_{\delta}^{1/2}\}^{-1} \{ I+\{(VW) \otimes \Sigma_{\delta}^{1/2}\}^{-1}\Sigma_z \{(VW) \otimes \Sigma_{\delta}^{1/2}\}^{-1}\}^{-1} \{V \otimes \Sigma_{\delta}^{1/2}\}^{-1} \bm{\omega}\right\} d\bm{v}\\
    &= \int_{\mathbb{R}_+^{T-1}} \frac{\prod_{t=1}^{T-1}\pi(\|\bm{z}_t\|_2^2 v_t ) /\pi (\|\bm{z}_t\|_2^2) }{| (W^{-1}\otimes I_L)\Sigma(W^{-1}\otimes I_L) + \{V^2\otimes \Sigma_{\delta} \}|^{1/2}} \nonumber \\
    &\qquad \exp\left\{ -\frac{1}{2}\bm{\omega}^{\top} \{V \otimes \Sigma_{\delta}^{1/2}\}^{-1} \{ I+\{(VW) \otimes \Sigma_{\delta}^{1/2}\}^{-1}\Sigma_z \{(VW) \otimes \Sigma_{\delta}^{1/2}\}^{-1}\}^{-1} \{V \otimes \Sigma_{\delta}^{1/2}\}^{-1} \bm{\omega}\right\}   \nonumber  \\ 
    &\qquad \left\{ -\{(VW) \otimes \Sigma_{\delta}^{1/2}\}^{-1} \{ I+\{(VW) \otimes \Sigma_{\delta}^{1/2}\}^{-1}\Sigma_z \{(VW) \otimes \Sigma_{\delta}^{1/2}\}^{-1}\}^{-1} \{V \otimes \Sigma_{\delta}^{1/2}\}^{-1} \bm{\omega}\right\}d\bm{v}  \nonumber  \\
    & \quad / \int_{\mathbb{R}_+^{T-1}}\frac{\prod_{t=1}^{T-1}\pi(\|\bm{z}_t\|_2^2 v_t ) /\pi (\|\bm{z}_t\|_2^2)}{| (W^{-1}\otimes I_L)\Sigma(W^{-1}\otimes I_L) + \{V^2\otimes \Sigma_{\delta} \}|^{1/2}} \nonumber \\
    &\qquad \exp\left\{ -\frac{1}{2}\bm{\omega}^{\top} \{V \otimes \Sigma_{\delta}^{1/2}\}^{-1} \{ I+\{(VW) \otimes \Sigma_{\delta}^{1/2}\}^{-1}\Sigma_z \{(VW) \otimes \Sigma_{\delta}^{1/2}\}^{-1}\}^{-1} \{V \otimes \Sigma_{\delta}^{1/2}\}^{-1} \bm{\omega}\right\} d\bm{v}\\
    &= \int_{\mathbb{R}_+^{T-1}} h(\bm{v}, z_{t^*}) \left\{ -\{(VW) \otimes \Sigma_{\delta}^{1/2}\}^{-1} \{ I+\{(VW) \otimes \Sigma_{\delta}^{1/2}\}^{-1}\Sigma_z \{(VW) \otimes \Sigma_{\delta}^{1/2}\}^{-1}\}^{-1}  \{V \otimes \Sigma_{\delta}^{1/2}\}^{-1} \bm{\omega}\right\}\\ 
    &\quad d\bm{v}/ \int_{\mathbb{R}_+^{T-1}}h(\bm{v}, z_{t^*})  d\bm{v}, 
\end{align*}where \begin{align*}
    h(\bm{v}, z_{t^*})  \equiv& \frac{\prod_{t=1}^{T-1}\pi(\|\bm{z}_t\|_2^2 v_t ) /\pi (\|\bm{z}_t\|_2^2)}{| (W^{-1}\otimes I_L)\Sigma(W^{-1}\otimes I_L) + \{V^2\otimes \Sigma_{\delta} \}|^{1/2}} \\ 
    &\exp\left\{ -\frac{1}{2}\bm{\omega}^{\top} \{V \otimes \Sigma_{\delta}^{1/2}\}^{-1} \{ I+\{(VW) \otimes \Sigma_{\delta}^{1/2}\}^{-1}\Sigma_z \{(VW) \otimes \Sigma_{\delta}^{1/2}\}^{-1}\}^{-1} \{V \otimes \Sigma_{\delta}^{1/2}\}^{-1} \bm{\omega}\right\}.
\end{align*}

Since there exists $\varepsilon_1<1,\varepsilon_2<1,M_1$ and $M_2$ such that $\varepsilon_1<\Sigma_z<M_1$ and $\varepsilon_2<\Sigma_{\delta}<M_2$, we have 
\begin{align*}
    \prod_{t=1}^{T-1} \frac{\pi(\|\bm{z}_t\|_2^2 v_t ) }{\pi (\|\bm{z}_t\|_2^2) }
    = \prod_{t=1}^{T-1} \frac{1+\|\bm{z}_t\|_2^4}{1+\|\bm{z}_t\|_2^4v_t^2}
    =\prod_{t\neq t^*}  \frac{1+\|\bm{z}_t\|_2^4}{1+\|\bm{z}_t\|_2^4v_t^2}\times \frac{2\|\bm{z}_{t^*}\|_2^4}{1+\|\bm{z}_{t^*}\|_2^4v_t^2},
\end{align*}
\begin{align*}
    \exp\left\{ -\frac{1}{2} \bm{\omega}^{\top} \{ M_2 (V^2\otimes I_L) + M_1 (W^{-2}\otimes I_L) \}^{-1} \bm{\omega}\right\}
    &= \exp\left\{ -\frac{1}{2} \sum_{t=1}^{T-1} \{ M_2 v_t^{2L} + M_1 \|\bm{z}_t\|_2^{-2L} \}^{-1} \bm{\omega}\right\}\\
    &=\prod_{t=1}^{T-1} \exp\left\{ -\frac{1}{2}  \frac{\|\bm{z}_t\|_2^{4L}}{M_2 v_t^{4L}\|\bm{z}_t\|_2^{2L}+M_1 } \right\},
\end{align*}
and \begin{align*}
    | (W^{-1}\otimes I_L)\Sigma(W^{-1}\otimes I_L) + V^2\otimes \Sigma_{\delta} |^{-1/2}
    &\le | (\varepsilon_1 W^{-2}+ \varepsilon_2 V^2 )\otimes I_L|^{-1/2}\\
    &= \prod_{t=1}^{T-1} \left( \frac{\|\bm{z}_t\|_2^4}{\varepsilon_1+\varepsilon_2\|\bm{z}_t\|_2^4 v_t^2 }\right)^{L/2} .
\end{align*}
Hence we obtain \begin{align*}
    h(\bm{v}, z_{t^*}) \le&  
    \left(\prod_{t\neq t^*}  \frac{\|\bm{z}_t\|_2^4}{\varepsilon_1+\varepsilon_2\|\bm{z}_t\|_2^4 v_t^2 } \frac{1+\|\bm{z}_t\|_2^4}{1+\|\bm{z}_t\|_2^4v_t^2}\exp\left\{ -\frac{1}{2}  \frac{\|\bm{z}_t\|_2^{4L}}{M_2 v_t^{2L}\|\bm{z}_t\|_2^{4L}+M_1 } \right\}\right)\\
    &\times    \left( \frac{\|\bm{z}_{t^*}\|_2^4}{\varepsilon_1+\varepsilon_2\|\bm{z}_{t^*}\|_2^4 v_{t^*}^2 }\right)^{L/2}
    \frac{2\|\bm{z}_{t^*}\|_2^4}{1+\|\bm{z}_{t^*}\|_2^4v_{t^*}^2}
    \exp\left\{ -\frac{1}{2}  \frac{\|\bm{z}_{t^*}\|_2^{4L}}{M_2 v_{t^*}^{2L}\|\bm{z}_{t^*}\|_2^{4L}+M_1 } \right\}.
\end{align*}
If $v_{t^*} >1$, it holds that\begin{align*}
    \left( \frac{\|\bm{z}_{t^*}\|_2^4}{\varepsilon_1+\varepsilon_2\|\bm{z}_{t^*}\|_2^4 v_{t^*}^2 }\right)^{L/2}
    \frac{2\|\bm{z}_{t^*}\|_2^4}{1+\|\bm{z}_{t^*}\|_2^4v_{t^*}^2}
    \exp\left\{ -\frac{1}{2}  \frac{\|\bm{z}_{t^*}\|_2^{4L}}{M_2 v_{t^*}^{2L}\|\bm{z}_{t^*}\|_2^{4L}+M_1 } \right\} 
    \le \left( \frac{1}{\varepsilon_2v^2_{t^*}} \right)^{L/2}\left( \frac{2}{v^2_{t^*}} \right)^{1/2},
\end{align*}otherwise, it holds that \begin{align*}
    &\left( \frac{\|\bm{z}_{t^*}\|_2^4}{\varepsilon_1+\varepsilon_2\|\bm{z}_{t^*}\|_2^4 v_{t^*}^2 }\right)^{L/2}
    \frac{2\|\bm{z}_{t^*}\|_2^4}{1+\|\bm{z}_{t^*}\|_2^4v_{t^*}^2}
    \exp\left\{ -\frac{1}{2}  \frac{\|\bm{z}_{t^*}\|_2^{4L}}{M_2 v_{t^*}^{2L}\|\bm{z}_{t^*}\|_2^{4L}+M_1 } \right\} \\
    &\le \left(  \frac{M_1}{\varepsilon_1} \frac{\|\bm{z}_{t^*}\|_2^4}{M_1+\varepsilon_2\|\bm{z}_{t^*}\|_2^4v_{t^*}^2}  \right)^{L/2} \left(\frac{2M_1\|\bm{z}_{t^*}\|_2^4}{M_1+\varepsilon_2\|\bm{z}_{t^*}\|_2^4v_{t^*}^2}  \right)
    \exp\left\{ -\frac{1}{2} \frac{\varepsilon_2}{M_2} \frac{\|\bm{z}_{t^*}\|_2^{4}}{\varepsilon_2v_{t^*}^{2}\|\bm{z}_{t^*}\|_2^{4}+M_1 } \right\} \\
    &\le \left(  \frac{M_1}{\varepsilon_1}  \right)^{L/2} 2M_1 \sup_{u\in (0,\infty)} u^{(L+2)/2}\exp\left(-\frac{\varepsilon_2}{M_2}\frac{-u}{2}\right)\equiv M_3<\infty
\end{align*}
Thus, we get\begin{align}
        h(\bm{v}, z_{t^*}) \le&  
    \left(\prod_{t\neq t^*}  \frac{\|\bm{z}_t\|_2^4}{\varepsilon_1+\varepsilon_2\|\bm{z}_t\|_2^4 v_t^2 } \frac{1+\|\bm{z}_t\|_2^4}{1+\|\bm{z}_t\|_2^4v_t^2}\exp\left\{ -\frac{1}{2}  \frac{\|\bm{z}_t\|_2^{4L}}{M_2 v_t^{2L}\|\bm{z}_t\|_2^{4L}+M_1 } \right\}\right)\nonumber \\
    &\times   \left(1_{\{v_{t^*} \le 1\}}M_3  +1_{\{v_{t^*} >1\}}\left( \frac{1}{\varepsilon_2v^2_{t^*}} \right)^{L/2}\left( \frac{2}{v^2_{t^*}} \right)^{1/2}\right).\label{eq:st3}
\end{align}
The right-hand side is integrable and does not depend on $\bm{z}_{t^*}$.

Next, we consider
\begin{align*}
    h(\bm{v}, z_{t^*}) \left\{ -\{(VW) \otimes \Sigma_{\delta}^{1/2}\}^{-1} \{ I+\{(VW) \otimes \Sigma_{\delta}^{1/2}\}^{-1}\Sigma_z \{(VW) \otimes \Sigma_{\delta}^{1/2}\}^{-1}\}^{-1}  \{V \otimes \Sigma_{\delta}^{1/2}\}^{-1} \bm{\omega}\right\}.
\end{align*}
For large $M_4$, it holds that \begin{align*}
    &\|\{(VW) \otimes \Sigma_{\delta}^{1/2}\}^{-1} \{ I+\{(VW) \otimes \Sigma_{\delta}^{1/2}\}^{-1}\Sigma_z \{(VW) \otimes \Sigma_{\delta}^{1/2}\}^{-1}\}^{-1}  \{V \otimes \Sigma_{\delta}^{1/2}\}^{-1} \bm{\omega} \|
    \\
    &\le \|W^{-1}\otimes I_L\|  \|\{ V^2 \otimes \Sigma_{\delta}  +\{W \otimes I_L\}^{-1}\Sigma_z \{W \otimes I_L\}^{-1} \}^{-1} \| \| \bm{\omega} \|\\
    &\le \sqrt{\sum_{t=1}^{T-1} \frac{L(T-1)}{\|\bm{z}_t\|_2^4} }\sum_{i=1}^{L(T-1)}\sum_{j=1}^{L(T-1)}
    |\bm{e}_i^{\top}\{ V^2 \otimes \Sigma_{\delta}  +\{W \otimes I_L\}^{-1}\Sigma_z \{W \otimes I_L\}^{-1} \}^{-1}\bm{e}_j  |
    \\
    &\le \sqrt{\sum_{t=1}^{T-1} \frac{L(T-1)}{\|\bm{z}_t\|_2^4} }\sum_{i=1}^{L(T-1)}\sum_{j=1}^{L(T-1)}
    \frac{|\bm{e}_i^{\top} \{(\varepsilon_1 W^{-2}+ \varepsilon_2 V^2 )\otimes I_L\}^{-1} \bm{e}_i  |+|\bm{e}_j^{\top} \{(\varepsilon_1 W^{-2}+ \varepsilon_2 V^2 )\otimes I_L\}^{-1} \bm{e}_j  |}{2}\\
    &\le \sqrt{\sum_{t=1}^{T-1} \frac{L(T-1)}{\|\bm{z}_t\|_2^4} } \frac{L}{2}\sum_{i=1}^{T-1}\sum_{j=1}^{T-1}
    \left\{  \left(\frac{\varepsilon_1}{\|\bm{z}_i\|_2^4} +\varepsilon_2v_i^2\right)^{-1} + \left(\frac{\varepsilon_1}{\|\bm{z}_j\|_2^4} +\varepsilon_2v_j^2\right)^{-1} \right\} \\
    &\le \sqrt{\sum_{t=1}^{T-1} \frac{L(T-1)}{\|\bm{z}_t\|_2^4} } L(T-1)\sum_{t=1}^{T-1}\frac{\|\bm{z}_t\|_2^4}{\varepsilon_1 + \varepsilon_2v_i^2\|\bm{z}_t\|_2^4 }\\
    &\le  M_4 \sum_{t=1}^{T-1}\frac{\|\bm{z}_t\|_2^4}{\varepsilon_1 + \varepsilon_2v_t^2\|\bm{z}_t\|_2^4 }.
\end{align*}
Thus, we have \begin{align*}
     &\|h(\bm{v}, z_{t^*}) \{(VW) \otimes \Sigma_{\delta}^{1/2}\}^{-1} \{ I+\{(VW) \otimes \Sigma_{\delta}^{1/2}\}^{-1}\Sigma_z \{(VW) \otimes \Sigma_{\delta}^{1/2}\}^{-1}\}^{-1}  \{V \otimes \Sigma_{\delta}^{1/2}\}^{-1} \bm{\omega} \| \\
     & \le M_4 \left(\sum_{t'\neq t^*} \frac{\|\bm{z}_t'\|_2^4}{\varepsilon_1 + \varepsilon_2v_{t'}^2\|\bm{z}_{t'}\|_2^4 } +\frac{\|\bm{z}_{t^*}\|_2^4}{\varepsilon_1 + \varepsilon_2v_{t^*}^2\|\bm{z}_{t^*}\|_2^4 } \right)\\
     &\qquad \times 
    \left(\prod_{t\neq t^*} \left( \frac{\|\bm{z}_t\|_2^4}{\varepsilon_1+\varepsilon_2\|\bm{z}_t\|_2^4 v_t^2 } \right)^{L/2}\frac{1+\|\bm{z}_t\|_2^4}{1+\|\bm{z}_t\|_2^4v_t^2}\exp\left\{ -\frac{1}{2}  \frac{\|\bm{z}_t\|_2^{4L}}{M_2 v_t^{2L}\|\bm{z}_t\|_2^{4L}+M_1 } \right\}\right)\\
     &\qquad \times    \left( \frac{\|\bm{z}_{t^*}\|_2^4}{\varepsilon_1+\varepsilon_2\|\bm{z}_{t^*}\|_2^4 v_{t^*}^2 }\right)^{L/2}
    \frac{2\|\bm{z}_{t^*}\|_2^4}{1+\|\bm{z}_{t^*}\|_2^4v_{t^*}^2}
    \exp\left\{ -\frac{1}{2}  \frac{\|\bm{z}_{t^*}\|_2^{4L}}{M_2 v_{t^*}^{2L}\|\bm{z}_{t^*}\|_2^{4L}+M_1 } \right\}\\
    & \le M_4 \sum_{t'\neq t^*}\frac{\|\bm{z}_{t'}\|_2^4}{\varepsilon_1 + \varepsilon_2v_{t'}^2\|\bm{z}_{t'}\|_2^4 } 
    \Big(\prod_{t\neq t^*}  \left( \frac{\|\bm{z}_t\|_2^4}{\varepsilon_1+\varepsilon_2\|\bm{z}_t\|_2^4 v_t^2 } \right)^{L/2+1_{\{t=t'\}}} \frac{1+\|\bm{z}_t\|_2^4}{1+\|\bm{z}_t\|_2^4v_t^2}\exp\left\{ -\frac{1}{2}  \frac{\|\bm{z}_t\|_2^{4L}}{M_2 v_t^{2L}\|\bm{z}_t\|_2^{4L}+M_1 } \right\}\Big)\\
     &\qquad \times    \left( \frac{\|\bm{z}_{t^*}\|_2^4}{\varepsilon_1+\varepsilon_2\|\bm{z}_{t^*}\|_2^4 v_{t^*}^2 }\right)^{L/2}
    \frac{2\|\bm{z}_{t^*}\|_2^4}{1+\|\bm{z}_{t^*}\|_2^4v_{t^*}^2}
    \exp\left\{ -\frac{1}{2}  \frac{\|\bm{z}_{t^*}\|_2^{4L}}{M_2 v_{t^*}^{2L}\|\bm{z}_{t^*}\|_2^{4L}+M_1 } \right\}\\
    &\quad+M_4    \left(\prod_{t\neq t^*}  \left(\frac{\|\bm{z}_t\|_2^4}{\varepsilon_1+\varepsilon_2\|\bm{z}_t\|_2^4 v_t^2 }\right)^{L/2} \frac{1+\|\bm{z}_t\|_2^4}{1+\|\bm{z}_t\|_2^4v_t^2}\exp\left\{ -\frac{1}{2}  \frac{\|\bm{z}_t\|_2^{4L}}{M_2 v_t^{2L}\|\bm{z}_t\|_2^{4L}+M_1 } \right\}\right)\\
     &\qquad \times    \left( \frac{\|\bm{z}_{t^*}\|_2^4}{\varepsilon_1+\varepsilon_2\|\bm{z}_{t^*}\|_2^4 v_{t^*}^2 }\right)^{L/2+1}
    \frac{2\|\bm{z}_{t^*}\|_2^4}{1+\|\bm{z}_{t^*}\|_2^4v_{t^*}^2}
    \exp\left\{ -\frac{1}{2}  \frac{\|\bm{z}_{t^*}\|_2^{4L}}{M_2 v_{t^*}^{2L}\|\bm{z}_{t^*}\|_2^{4L}+M_1 } \right\}\\
     & \le M_4 \sum_{t'\neq t^*}\frac{\|\bm{z}_{t'}\|_2^4}{\varepsilon_1 + \varepsilon_2v_{t'}^2\|\bm{z}_{t'}\|_2^4 } 
    \Big(\prod_{t\neq t^*}  \left( \frac{\|\bm{z}_t\|_2^4}{\varepsilon_1+\varepsilon_2\|\bm{z}_t\|_2^4 v_t^2 } \right)^{L/2+1_{\{t=t'\}}} \frac{1+\|\bm{z}_t\|_2^4}{1+\|\bm{z}_t\|_2^4v_t^2}\exp\left\{ -\frac{1}{2}  \frac{\|\bm{z}_t\|_2^{4L}}{M_2 v_t^{2L}\|\bm{z}_t\|_2^{4L}+M_1 } \right\}\Big)\\
     &\qquad \times \left(1_{\{v_{t^*} \le 1\}}M_3  +1_{\{v_{t^*} >1\}}\left( \frac{1}{\varepsilon_2v^2_{t^*}} \right)^{L/2}\left( \frac{2}{v^2_{t^*}} \right)^{1/2}\right)\\
    &\quad+M_4    \left(\prod_{t\neq t^*}  \left(\frac{\|\bm{z}_t\|_2^4}{\varepsilon_1+\varepsilon_2\|\bm{z}_t\|_2^4 v_t^2 }\right)^{L/2} \frac{1+\|\bm{z}_t\|_2^4}{1+\|\bm{z}_t\|_2^4v_t^2}\exp\left\{ -\frac{1}{2}  \frac{\|\bm{z}_t\|_2^{4L}}{M_2 v_t^{2L}\|\bm{z}_t\|_2^{4L}+M_1 } \right\}\right)\\
     &\qquad \times    \left(1_{\{v_{t^*} \le 1\}}M_5  +1_{\{v_{t^*} >1\}}\left( \frac{1}{\varepsilon_2v^2_{t^*}} \right)^{L/2}\left( \frac{2}{v^2_{t^*}} \right)^{1/2}\right),
\end{align*}
for some $M_5>0$. The last one is integrable and independent of $\bm{z}_{t^*}$. Combining this and (\ref{eq:st3}) shows the existence of an integrable function that dominates the integrand in (\ref{eq:st2}).

\textbf{Step 4: Conclusion.} 
Step 3 shows the dominated convergence theorem is applicable to (\ref{eq:st2}). Then, from the dominated convergence theorem and Step 2,  we obtain \begin{align*}
    \frac{1}{m(\bm{z})}\frac{\partial m(\bm{z})}{\partial\bm{z}} \sim -\Sigma_z \frac{\int h(\bm{v}, \bm{z}_{t^*})\{-A(v_{-i}) (I+A(v_{-i})\Sigma A(v_{-i}))^{-1} \{V \otimes \Sigma_{\delta}^{1/2}\}^{-1} \bm{\omega}\} d\bm{v}}{\int h(\bm{v}, z_{t^*})d\bm{v}},
\end{align*}
as $\bm{z}_{t^*}\to\infty$. Therefore ${\rm E}[\bm{\eta}|\bm{z}] - \bm{z} $ is bounded and this concludes the proof.

\end{proof}

\bibliographystyle{chicago}
\bibliography{ref}

\end{document}